\documentclass[reqno]{amsart}

\usepackage[utf8]{inputenc}

\usepackage{fullpage}
\usepackage{appendix}
\usepackage{newtxtext}
\usepackage{enumitem}

\usepackage{amssymb}
\usepackage{amsmath}
\usepackage{amsthm}
\usepackage{amsopn}

\usepackage{mathrsfs}
\usepackage{mdframed}
\usepackage{tikz-cd}
\usepackage{subfig}
\usepackage{graphicx}
\usepackage{esint}
\usepackage{cancel}
\usepackage{thmtools, thm-restate}
\usepackage{hyperref}

\usepackage{verbatim}
\usepackage{mathtools}
\usepackage{fullpage}
\usepackage{dsfont}

\usepackage{physics}
\usepackage{setspace}
\usepackage{epstopdf}
\usepackage[T1]{fontenc}
\usepackage{graphicx}

\usepackage{scalerel,stackengine}
\usepackage{ulem}

\usepackage{soul}

\usepackage[numbers]{natbib}

\usepackage{marginnote}

\newcommand{\dx}{\mathrm{d}x}
\newcommand{\dy}{\mathrm{d}y}

\newcommand{\ddp}{\mathrm{d}p}
\newcommand{\dq}{\mathrm{d}q}

\newcommand{\dw}{\mathrm{d}w}
\newcommand{\du}{\mathrm{d}u}

\newcommand{\dt}{\mathrm{d}t}
\newcommand{\ds}{\mathrm{d}s}

\newcommand{\deta}{\mathrm{d}\eta}

\newcommand{\bR}{\mathbb{R}}
\newcommand{\bN}{\mathbb{N}}

\newcommand{\bC}{\mathbb{C}}

\newcommand{\cR}{\mathcal{R}}

\newcommand{\cU}{\mathcal{U}}
\newcommand{\cN}{\mathcal{N}}

\newcommand{\rms}{\mathrm{s}}
\newcommand{\rmu}{\mathrm{u}}
\newcommand{\rmv}{\mathrm{v}}
\newcommand{\rmm}{\mathrm{m}}
\newcommand{\rme}{\mathrm{e}}

\newcommand{\ii}{\mathrm{i}}

\newcommand{\mytodo}[2][]{{%
		\let\marginpar\marginnote
		\reversemarginpar
		\renewcommand{\baselinestretch}{0.8}%
		\todo[#1]{#2}}}

\theoremstyle{remark}
\newtheorem{Remark}{Remark}[section]
\theoremstyle{plain}
\newtheorem{Theorem}{Theorem}[section]
\newtheorem{Lemma}{Lemma}[section]

\newtheorem{Proposition}{Proposition}[section]
\newtheorem{assump}{Assumption}

\DeclareRobustCommand{\rchi}{{\mathpalette\irchi\relax}}
\newcommand{\irchi}[2]{\raisebox{\depth}{$#1\chi$}} % inner command, used by \rchi
\newcommand\numberthis{\addtocounter{equation}{1}\tag{\theequation}}

\DeclareMathOperator{\sign}{sign}

\DeclareMathOperator{\supp}{supp}

\def \qed {\hfill \vrule height6pt width 6pt depth 0pt}

\newcommand{\R}{\bR}
\newcommand{\N}{\bN}

\newcommand{\p}{\partial}

\newcommand{\hsnorm}[1]{\norm{#1}_{\mathrm{HS}}}
\newcommand{\opnorm}[1]{\norm{#1}_{\mathrm{op}}}
\newcommand{\trnorm}[1]{\norm{#1}_{\Tr}}

\usepackage{ulem}

\usepackage{natbib}



\allowdisplaybreaks

\numberwithin{equation}{section}

\stackMath
\newcommand\reallywidehat[1]{%
	\savestack{\tmpbox}{\stretchto{%
			\scaleto{%
				\scalerel*[\widthof{\ensuremath{#1}}]{\kern-.6pt\bigwedge\kern-.6pt}%
				{\rule[-\textheight/2]{1ex}{\textheight}}%WIDTH-LIMITED BIG WEDGE
			}{\textheight}%
		}{0.5ex}}%
	\stackon[1pt]{#1}{\tmpbox}%
}
\renewcommand{\bar}{\overline}

\renewcommand{\tilde}{\widetilde}
\renewcommand{\leq}{\leqslant}
\renewcommand{\geq}{\geqslant}

\title{A mixed-norm estimate of the two-particle reduced density matrix of many-body Schrödinger dynamics for deriving the Vlasov equation}

\author{Li Chen}
\address[L.~Chen]{Institut für Mathematik, Universität Mannheim}
\email{chen@math.uni-mannheim.de}

\author{Jinyeop Lee}
\address[J.~Lee]{Mathematisches Institut, Ludwig-Maximilians-Universität München}
\email{lee@math.lmu.de}

\author{Yue Li}
\address[Y.~Li]{Department of Mathematics, Nanjing University}
\email{liyue2011008@163.com}

\author{Matthew Liew}
\address[M.~Liew]{Institut für Mathematik, Universität Mannheim}
\email{mliew@mail.uni-mannheim.de}

\begin{document}
	
	\maketitle
	
	\begin{abstract}
		\vspace{1em}
		We re-examine the combined semi-classical and mean-field limit in the $N$-body fermionic Schrödinger equation with pure state initial data using the Husimi measure framework. The Husimi measure equation involves three residue types: kinetic, semiclassical, and mean-field. 
		The main result of this paper is to provide better estimates for the kinetic and mean-field residue than those in  \cite{Chen2021JSP}.
		Especially, the estimate for the mean-field residue is shown to be smaller than the semiclassical residue by a mixed-norm estimate of the two-particle reduced density matrix factorization.
		Our analysis also updates the oscillation estimate parts in the residual term estimates appeared in \cite{Chen2021JSP}.
	
		\vspace{1em}
	\end{abstract}

	\normalem
	
	\section{Introduction}
	
	In this article, we consider the following $N$-particle mean-field Schr\"odinger equation
	\begin{equation}\label{eq:Schrodinger_0}
		\begin{split}
			\ii \hbar \partial_t\, \psi_{N,t} &=-\frac{\hbar^2}{2} \sum_{j=1}^{N}  \Delta_{x_j} \psi_{N,t} + \frac{1}{2N} \sum_{i\neq j}^{N} V(x_i-x_j) \psi_{N,t}\\
			\psi_{N,0} &= \frac{1}{\sqrt{N!}} \det{\rme_j (x_i)}_{i,j=1}^{N},
		\end{split}
	\end{equation}
	where $\{\rme_j\}_{j=1}^N$ is a family of orthonormal basis in $L^2(\bR^3)$ and $\Delta_{x_j}$ is the Laplacian on $j$-th particle.
	The initial data in \eqref{eq:Schrodinger_0} is in the form of a Slater determinant, which stays in the antisymmetric subspace $L^2_a(\R^{3N})$ of $L^2(\bR^3)$ with $\norm{\psi_{N,0}}_2=1$, where
	\[
	L^2_a(\R^{3N}) := \big\{ \psi_N \in L^2(\R^{3N}) :\ \psi_N(x_{\pi(1)},\dots,x_{\pi(N)}) = (\sign{\pi})  \psi_N(x_1, \dots, x_N)\ \text{for all } \pi \in S_N\big\}.
	\]
	In the above formulation, $S_N$ is the {{symmetric}} group.
	
	Note that the number of the terms for interaction is of order $N^2$. Hence, with the mean-field constant $1/N$ in front of the interaction, we could think that the size of interaction energy is of order $N$.
	Since we are interested in the regime where the size of kinetic energy is similar to the size of interaction energy, 
	we have from Tomas-Fermi theory that
	$
	\hbar^2 N^{5/3} = N.
	$
	This gives
	$
	\hbar = N^{-1/3}.
	$
	For more details, we refer to \cite{BACH20161, benedikter2014mean, benedikter2022effective}.
	
	As it is difficult to solve the Schr\"odinger equation in \eqref{eq:Schrodinger_0} numerically when the number of particle $N$ is large, we aim to derive its corresponding effective evolution equation.
	In fact, we consider the $k$-particle reduced density matrix where its corresponding integral kernel is given by
	\begin{equation}\label{eq:gammak}
		\begin{aligned}
			&\gamma_{N,t}^{(k)} (x_1,\dots,x_k;y_1,\dots,y_k)\\
			&= \frac{N!}{(N-k)!} \int \mathrm{d}{x_{k+1}} \dots \mathrm{d}{x_N}\, \overline{\psi_{N,t}(y_1,\dots,y_k,x_{k+1},\dots,x_N)} \psi_{N,t}(x_1, \dots, x_k,x_{k+1},\dots,x_N),
		\end{aligned}
	\end{equation}
	where $1\leq k \leq N$. Moreover, we denote the expectation of the one-particle observable as follows,
	\[
	\Tr O\gamma^{(1)}_{N,t} =\left< \psi_{N,t}, O \psi_{N,t} \right> = \int \mathrm{d}{x_1} \dots \mathrm{d}{x_N}\, \overline{\psi_{N,t}({x_1},x_{2},\dots,x_N)} \big(O \psi_{N,t}\big)(x_1,x_{2},\dots,x_N).
	\]
	The one-particle reduce density matrix of the initial data given in \eqref{eq:Schrodinger_0} is
	$
	\omega_N = \sum_{j=1}^N \dyad{\rme_j}{\rme_j},
	$
	where its corresponding integral kernel is $\omega_N(x;y) = \sum_{j=1}^N \overline{\rme_j(y)} \rme_j(x)$.
	
	\subsubsection{Short review of mean-field limit, $N\to \infty$}
	
	It is well known that the Hartree-Fock equation 
	\begin{equation}\label{eq:hartree-fock}
		\begin{split}
			\ii \hbar \partial_t \, \omega_{N,t} &= \big[-\hbar^2 \Delta + (V*\rho_{N,t}) - X_t, \omega_{N,t}\big],\\
			\rho_{N,t} &= \frac{1}{N} \omega_{N,t}(x;x) \\
			X_t &= \frac{1}{N} V(x-y) \omega_{N,t}(x;y)\\
		\end{split}
	\end{equation}
	is used to approximate the Schr\"odinger equation in the mean-field limit.
	Here we use the conventional notation $[A,B]:=AB-BA$ for commutator of operators.
	
	The mean field limit result for fixed $\hbar$ has been given in \cite{ELGART20041241} for short time.
	Under the scaling $\hbar = N^{-1/3}$, the rates of convergence in the trace norm and the Hilbert-Schmidt norm are obtained for arbitrary given time in \cite{benedikter2014mean} when the initial data is an approximation of the Slater determinant. Later on,  the case with mixed state initial data has been considered in \cite{benediktermixed,benedikter2014rel}. Furthermore, for Coulomb and Riesz potentials, the rate of convergence is obtained in \cite{Porta2017,saffirio2017mean}. We refer more references on this topic to \cite{Frohlich2011,petrat2014derivation,Petrat2017,Petrat2016ANM,benedikter2022effective} and the references therein.
	
	\subsubsection{Short review of semi-classical limit, $\hbar\to0$}
	
	The Vlasov equation can be obtained via semi-classical limit of the Hartree or the Hartree--Fock equations. It has been first investigated in \cite{Lions1993} by using Wigner measure for smooth potentials.
	Recently, the rates of convergence in the trace norm as well as the Hilbert-Schmidt norm has been studied in \cite{benedikter2016hartree} with {{regularity assumptions}} on the mixed state initial data and a class of regular potentials.
	The $k$-particle Wigner measure used in the \cite{Lions1993,benedikter2016hartree} reads
	\begin{equation}\label{eq:wigner_measure}
		\begin{aligned}
			&W^{(k)}_{N,t}(x_1,p_1, \dots, x_k, p_k) \\
			&:= \binom{N}{k}^{-1} \int (\dd{y})^{\otimes k} \; 
			\gamma_{N,t}^{(k)}\left(x_1 + \frac{\hbar}{2}y_1, \dots, x_k + \frac{\hbar}{2}y_k; x_1 - \frac{\hbar}{2}y_1, \dots, x_k - \frac{\hbar}{2}y_k \right)e^{-\ii \sum_{i=1}^k p_i \cdot y_i},
		\end{aligned}
	\end{equation}
	Some of the recent developments in the semi-classical limit are the following: One can find results 
	for the inverse power law potential in \cite{Saffirio2019}, for the rate of convergence in the Schatten norm in \cite{lafleche2020strong}, for the Coulomb potential and mixed states in \cite{saffirio2019hartree}, and for the convergence in the Wasserstein distance in \cite{Lafleche2019GlobalSL,Lafleche2019PropagationOM}.
	Relativistic fermionic system has been studied in \cite{Dietler2018}. Further analyses of the semi-classical limit can be found in \cite{amour2013classical,amour2013,Athanassoulis2011,Gasser1998, Markowich1993,benedikter2022effective}.
	
	\subsubsection{Combined Mean-Field and semi-classical Limits}
	Narnhofer and Sewell, and Spohn independently derived  Vlasov equation~\eqref{eq:vlasov} from the $N$-body Schrödinger equation \eqref{eq:Schrodinger_0} with $\hbar = N^{-1/3}$, in \cite{Narnhofer1981,Spohn1981}.
	Without assuming $\hbar = N^{-1/3}$, a rate of convergence was obtained in \cite{graffi2003mean} in a weak formulation.
	The rate of convergence of the combined limits was studied in \cite{Golse2017,Golse2021,golse:hal-01334365} by using the Wasserstein (pseudo-)distance.
	Under a generalized Husimi measure framework, the authors in \cite{Chen2021JSP} obtained the convergence for regular potentials. 
	Recently, the combined limit for the singular potential case with regular mixed state initial data was obtained in \cite{Chong2021}.
	
	It is well-known that the Wigner measure in \eqref{eq:wigner_measure} is not a (proper) probability measure, as there might be some point having negative sign. (We refer, e.g.,  ~\cite{doi:10.1063/1.531326,Hudson1974,kenfack2004negativity,PhysRevA.79.062302,doi:10.1063/1.525607} for further references on Wigner measure.)
	It has been shown that the Husimi measure, the convolution of the Wigner measure with a Gaussian function, is a nonnegative probability measure
	\cite{Combescure2012,Fournais2018,Zhang2008}.
	In particular, from \cite[p.21]{Fournais2018}, given a specific Gaussian coherent state, the relation between the Husimi measure and Wigner measure is given by the following convolution: for any $1 \leq k \leq N$,
	\begin{equation}\label{eq:gaussian_m_vs_W}
		\mathfrak{m}_{N,t}^{(k)}  := \frac{N(N-1)\cdots (N-k+1)}{N^k} W^{(k)}_{N,t} * \mathcal{G}^\hbar,
	\end{equation}
	where
	\begin{equation*}
		\mathcal{G}^\hbar (q_1,p_1,\dots, q_k,p_k) := \frac{1}{(\pi \hbar)^{3k} }\exp \Big(-\frac{ \sum_{j=1}^k {q_j}^2+ {p_j}^2}{\hbar} \Big).
	\end{equation*}
	The Wigner transform of Hartree (or Hartree-Fock) equation shares the structure of Vlasov equation, see \cite[Eq. (6.15)]{Benedikter2016book} for example. 

In this paper, following the ideas in \cite{Chen2021JSP}, we study the equation for Husimi measure. Instead of using the classical definition of Husimi measure in \eqref{eq:gaussian_m_vs_W}, we consider the following generalized  $k$-particle Husimi measure, which is given for example in \cite{Fournais2018}:
For any $p,q \in \R^3$ and $\psi_{N,t} \in L^2_a(\R^{3N})$, the $k$-particle Husimi measure is given by
\begin{equation}\label{eq:husimi_def_1}
	m_{N,t}^{(k)} (q_1,p_1,\dots, q_k, p_k) = \langle \psi_{N,t}, a^*(f^\hbar_{q_1, p_1}) \cdots a^*(f^\hbar_{q_k, p_k}) a(f^\hbar_{q_k, p_k}) \cdots  a(f^\hbar_{q_1, p_1}) \psi_{N,t} \rangle.
\end{equation}
Here $a^*(f^\hbar_{q,p})$ and $a(f^\hbar_{q,p})$ are standard creation- and annihilation-operator respectively\footnote{Definitions of creation- and annihilation-operator is provided in Appendix \ref{sec:FockSpace}.} with respect to the coherent state $f^\hbar_{q,p}$ given by
\begin{equation*}\label{eq:def_coherent}
	f^{\hbar}_{q, p} (y) := \hbar^{-\frac{3}{4}} f \left(\frac{y-q}{\sqrt{\hbar}} \right) e^{\frac{\ii}{\hbar} p \cdot y },
\end{equation*}
where $f$ is any given real-valued function satisfying $\norm{f}_2 =1$.

\begin{Remark}
	As stated in \cite{Fournais2018}, the $k$-particle Husimi measure $m^{(k)}_{N,t}$ describes how many fermions are within the $k$-semi-classical boxes of length $\sqrt{\hbar}$ centered at the phase-spaces $(q_1,p_1), \dots, (q_k,p_k)$.
\end{Remark}

\begin{Remark}
	If $f (x)= \pi^{-3/4} e^{-|x|^2/2}$, \cite{Combescure2012} shows that the $k$-particle Husimi measure $m_{N,t}^{(k)}$ coincides with the $\mathfrak{m}^{(k)}_{N,t}$ in \eqref{eq:gaussian_m_vs_W}.
\end{Remark}

\subsection{Main Result}
{Let $\psi_{N,t}$ be the solution to the Schr\"odinger equation in \eqref{eq:Schrodinger_0} and denote the one-particle Husimi measure of it by $m_{N,t} := m_{N,t}^{(1)}$.} From  \cite[Proposition 2.1]{Chen2021JSP}, we obtain the following identity:
\begin{equation} \label{eq:chen2021_representation}
	\begin{aligned}
		\p_t m_{N,t}(q,p) &+ p \cdot \nabla_{q} m_{N,t}(q,p) - \nabla_q\cdot\left(\hbar\, \Im \langle \nabla_{q} a (f^\hbar_{q,p}) \psi_{N,t}, a (f^\hbar_{q,p}) \psi_{N,t} \rangle\right)\\
		&= \frac{1}{(2\pi)^3}\nabla_p\cdot\int \dw_1 \du_1 \dw_2\du_2 \dq_2 \ddp_2\,    \left(  f^\hbar_{q,p}(w)  \overline{f^\hbar_{q,p}(u)} \right)^{\otimes 2}\\
		& \hspace{1.5cm}\int_0^1 \ds\ \nabla V_N\big( su_1+ (1-s)w_1-w_2 \big)  \gamma_{N,t}^{(2)}(u_1,u_2;w_1,w_2),
	\end{aligned}
\end{equation}
where
\[
\left(  f^\hbar_{q,p}(w_{1})  \overline{f^\hbar_{q,p}(u_{1})} \right)^{\otimes 2} :=  f^\hbar_{q,p}(w_1)  \overline{f^\hbar_{q,p}(u_1)}  f^\hbar_{q_2,p_2}(w_2)  \overline{f^\hbar_{q_2,p_2}(u_2)}.
\]

\begin{Remark}
	The two-particle reduced density matrix in \eqref{eq:chen2021_representation} {is given in} \eqref{eq:gammak}.
\end{Remark}

Our aim is therefore to obtain the convergence from $m_{N,t}$, in weak sense, to the solution of Vlasov equation $m_t$ as follows:
\begin{equation}\label{eq:vlasov}
	\begin{split}
		\partial_t m_{t}(q,p) &= -p \cdot \nabla_q m_t(q,p) + \nabla_q \big(V * \varrho_{t}\big)(q) \cdot \nabla_p  m_t(q,p),\\
		m_{t}(q,p)\big|_{t=0} &= m_0 (q,p),
	\end{split}
\end{equation}
where {$\varrho_{t} = \int \ddp\,  m_t(q,p)$ and the initial data $m_0 (q,p)$ is the one-particle Husimi measure of $\psi_{N,0}$ given in \eqref{eq:Schrodinger_0}.}

{The following assumptions are needed in this paper.}
\begin{assump}\label{assumptions}\phantom{text}
	\begin{enumerate}
		\item (Interaction potential) { $V\in L^1(\R^3)$ and $V(-x) = V(x)$. Furthermore it holds} $\int \mathrm{d}{p} (1+|p|^2) |\widehat{V}(p)| < \infty$, where $\widehat{V}$ is its Fourier transform of $V$.
		
		\item (Coherent state) $f \in H^1(\R^3)\cap L^\infty(\R^3)$ satisfies $\norm{f}_2 = 1$, and has compact support in $B_{R_1}$ {for a given $R_1>0$}.
		
		\item (Initial data) {$m_{N}$ converges weakly to $m_0$ in  $L^1(\R^3)$}. Furthermore, it satisfies
		\begin{equation}\label{eq:inip2q_ben}
			\int \dq \ddp\, (|p|^2 + |q|) m_{N}(q,p) <\infty
		\end{equation}
		{{uniformly for all $N$.}}
		
		\item (Initial data) { $\omega_N$, the one-particle density matrix of $\psi_{N,0}$} satisfies
		\begin{equation}\label{eq:assum_semi-classical}
			\begin{aligned}
				\sup_{p\in\bR^{3}}\frac{1}{1+|p|}\trnorm{[e^{\ii p\cdot x}, \omega_{N}]} &\leq CN\hbar,\\
				\trnorm{[\hbar \nabla, \omega_{N}]} &\leq CN\hbar,
			\end{aligned}
		\end{equation}
		where $\trnorm{\cdot}$ is the trace norm.
	\end{enumerate}
\end{assump}

\begin{Remark}
	The assumptions in \eqref{eq:assum_semi-classical} can be explained by the nature of the semi-classical structure. More details can be found in \cite{benedikter2014mean} where mean field limit has been studied.
\end{Remark}

With the assumptions presented above, our final goal is to obtain the following theorem:

\begin{Theorem} \label{thm:main_ben}	
	Let $m_{N,t}$ be the one-particle Husimi measure defined in \eqref{eq:husimi_def_1} with $\psi_{N,t}$ the solution of Schr\"odinger equation \eqref{eq:Schrodinger_0}, and suppose the aforementioned assumptions hold and $m_{t}$ is the solution of Vlasov equation in \eqref{eq:vlasov}. Then, for any given $T>0$, $(m_{N,t})_{N\in\bN}$ converges to $m_t$ weakly (*) in $L^p ((0,T)\times \R^3 \times \R^3)$ for arbitrary 
	$1\leq p \leq \infty$.
\end{Theorem}

\begin{Remark}
	As a consequence, \cite[Theorem 6.9]{Villani2003} implies, for $t\geq 0$, that $W_1 (m_{N,t}, m_t ) \to 0$ as $ N \to \infty$, where $W_1$ is 1-Weisserstein distance.
\end{Remark}

For convenience, we use the Fock space formalism, which will be briefly introduced in Appendix \ref{sec:FockSpace}.
By using Husimi transform given in \eqref{eq:chen2021_representation},\ 
the Schr\"odinger equation for $\psi_{N,t}$ can be rewritten into the Vlasov type equation for $m_{N,t}$  with residual terms. More precisely, from the computations in \cite{Chen2021JSP}, we have
\begin{equation} \label{eq:BBGKY_k1_ben}
	\begin{aligned}
		&\p_t m_{N,t}(q,p) + p \cdot \nabla_{q} m_{N,t}(q,p)\\
		&= \frac{1}{(2\pi)^3}\nabla_{p} \cdot  \int \dd{q_2} \nabla V(q- q_2) \varrho_{N,t}(q_2) m_{N,t}(q,p)
		+ \nabla_{q}\cdot \widetilde{\cR} +\nabla_{p}\cdot \cR_\rms +\nabla_{p}\cdot \cR_\rmm ,
	\end{aligned}
\end{equation}
where $\varrho_{N,t}(q):= \int \dd{p} m_{N,t}(q,p)$, the kinetic residue  $\widetilde{\cR}$, the semi-classical residue $\cR_\rms$ and the mean-field residue $\cR_\rmm$ are given by
\begin{equation}\label{eq:bbgky_remainder_1_ben}
	\begin{aligned}
		\tilde{\cR}_{\phantom{1}} &:=   \hbar \Im \left< \nabla_{q} a (f^\hbar_{q,p}) \Psi_{N,t}, a (f^\hbar_{q,p}) \Psi_{N,t} \right>,
		\\
		{\cR_\rms} &:= \frac{1}{(2\pi)^3}  \int \dw_1 \du_1 \dw_2\du_2 \dq_2 \ddp_2\,    \left(  f^\hbar_{q,p}(w)  \overline{f^\hbar_{q,p}(u)} \right)^{\otimes 2}\\
		&
		\qquad \qquad  \bigg[\int_0^1 \ds\ \nabla V\big( su_1+ (1-s)w_1-w_2 \big)   - \nabla V(q-q_2) \bigg]\gamma_{N,t}^{(2)}(u_1,u_2;w_1,w_2)  ,
		\\
		{\cR_\rmm} &:= 	\frac{1}{(2\pi)^3}  \int \dw_1 \du_1 \dw_2\du_2 \dq_2 \ddp_2\,    \left(  f^\hbar_{q,p}(w)  \overline{f^\hbar_{q,p}(u)} \right)^{\otimes 2}\\
		&
		\qquad  \nabla V(q-q_2) \bigg[ \gamma_{N,t}^{(2)}(u_1,u_2;w_1,w_2) -\gamma_{N,t}^{(1)}(u_1;w_1)\gamma_{N,t}^{(1)}(u_2;w_2)\bigg].
	\end{aligned}
\end{equation}

\begin{Remark}
The three terms - kinetic residue $\widetilde{\cR}$, semiclassical residue $\cR_\rms$, and mean-field residue $\cR_\rmm$ - arise due to the following reasons:
The kinetic residue $\widetilde{\cR}$ is bounded by the kinetic energy estimate. This term's name reflects its connection to the kinetic energy used in establishing the bound.
The semiclassical residue $\cR_\rms$ is proven to be small by oscillatory integrals that appear due to the Husimi transform. Similarly, the estimate of mean-field residue $\cR_\rmm$ is transformed into the factorization of two-particle reduced density matrix, which is a characteristic of mean-field behavior.
Thus, its name denotes its association with mean-field properties.
\end{Remark}

Under the assumptions of Theorem \ref{thm:main_ben},
	the estimate of the kinetic residue  $\tilde{\cR}$ can be obtained exactly the same as in \cite{Chen2021JSP}.
	(In the updated arXiv version of \cite{Chen2021JSP}, 
	the oscillation estimates have been corrected, with which
	the estimates for the residue terms from BBGKY hierarchy as well as the main result still hold true.)
	In this paper we give a better estimate for $\tilde{\cR}$ in the following proposition.
	\begin{Proposition}\label{prop:tildeR}
		Under the assumptions 1,2, and 3 in  \ref{assumptions}, the following estimate holds
		\begin{align}
			\norm{\int\ddp\,\big|\widetilde{\cR}(p,\cdot)\big|}_{L^{\frac{5}{4}}(\bR^3)}\leq C\hbar^{\frac{1}{2}}.
		\end{align}
As a consequence, the following holds:
\begin{align}
\bigg|\int\dd{q}\ddp\,\varphi(q)\phi(p) \nabla_{q}\cdot \tilde{\cR}(q,p)\bigg| \leq C\hbar^{\frac{1}{2}}.
\end{align}
where  $C$ depends on $\varphi, \phi, f$.
\end{Proposition}
The proof will be listed in Section \ref{sec:kinetic_R}.

In addition, we obtain the following two propositions for the other two residues:
\begin{Proposition}\label{prop:est_R_1} Under the assumptions 1,2, and 3 in  \ref{assumptions} and let $\phi, \varphi$ be test functions, then the following inequality holds:
\begin{equation}
\begin{aligned}
	\bigg|\int \dq \ddp\, \varphi(q) \phi(p) \nabla_{p} \cdot \cR_\rms(q,p) \bigg| 	\leq C \hbar^{\frac{1}{2}+3(\alpha_2 - 1)} ,
\end{aligned}
\end{equation}
where $\frac{1}{2} < \alpha_2 < 1$, and $C$ depends on $\varphi, \phi, f$.
\end{Proposition}
Detailed proof will be given in Section \ref{sec:semi-classical_R}.

\begin{Proposition}\label{prop:est_mean-field_osc} Assuming  \ref{assumptions},
	let $\varphi, \phi \in C_0^\infty(\R^3)$ be test functions, then the following inequality holds:
	\begin{equation}
		\bigg|\int \dq \ddp\, \varphi(q)\phi(p)  \nabla_{p} \cdot \cR_\rmm(q,p) \bigg| \leq {C} \hbar^{\frac{3}{2}(\alpha_1 {-} \frac{1}{2}) + \frac{3}{2}}
	\end{equation}
	where $\frac{1}{2} < \alpha_1 < 1$, $C$ depends on $\phi, \varphi, f$.
\end{Proposition}
Detailed proof will be given in Section \ref{sec:mean_field_residue}. 

By using a similar idea in the estimate of the semi-classical residue term, as shown in Lemma \ref{lem:est_mean-field_osc}, the estimate of the mean-field residue term can be reduced to the corresponding quantities involving
\begin{equation}\label{eq:intro_insert}
	\begin{aligned}
		&\gamma_{N,t}^{(2)}(u_1,u_2;w_1,w_2)  - \gamma_{N,t}^{(1)}(u_1;w_1)\gamma_{N,t}^{(1)}(u_2;w_2) \\
		&
		=  \gamma_{N,t}^{(2)}(u_1,u_2;w_1,w_2) - \omega_{N,t}(u_1;w_1) \omega_{N,t}(u_2;w_2)\\
		&
		\quad + \big[ \omega_{N,t}(u_1;w_1)  -  \gamma^{(1)}_{N,t}(u_1;w_1)\big] \omega_{N,t}(u_2;w_2)\\
		&
		\quad + \gamma^{(1)}_{N,t}(u_1;w_1) \big[\omega_{N,t}(u_2;w_2) -  \gamma_{N,t}^{(1)}(u_2;w_2)\big]\\
		&
		=: T_1 + T_2 +T_3,
	\end{aligned}
\end{equation}
where $\omega_{N,t}$ is the solution to Hartree-Fock equation.
The terms with {$T_2$ and $T_3$ can be estimated by the trace norm and Hilbert-Schmidt norm of $\gamma_{N,t}^{(1)} - \omega_{N,t}$, respectively.} To control $T_1$, we provide a bound of the following mixed norm estimate for two particle density matrix
\begin{equation}\label{eq:mixnorm}
	\left(\int \dw_1\du_1\ \left[ \int \mathrm{d}{w_2} \left|\gamma_{N,t}^{(2)}(w_1, w_2; u_1, w_2)  - \omega_{N,t}(w_1;u_1) \omega_{N,t}(w_2;w_2)  \right| \right]^2   \right)^\frac{1}{2}.
\end{equation}

In \cite{benedikter2014mean}, the convergence with respect to the trace norm and Hilbert-Schmidt norm of the difference between $\gamma_{N,t}^{(k)}$ and $\omega_{N,t}^{(k)}$ are obtained separately with the help of Wick's theorem for $k \geq 2$. {However, in the current framework the mixed norm estimate as listed in \eqref{eq:mixnorm} needs extra efforts. We trace the strategies given in \cite{benedikter2014mean} to reduce it to the estimate of the expectation of the number operator $\cN$ along the quantum fluctuation}.

{The above estimates show that} in the sense of distribution $\cR_\rms\sim \hbar^{\frac{1}{2}-}$ and $\cR_\rmm\sim \hbar^{\frac{9}{4}-}\sim N^{-\frac{3}{4}-}$, from which one can observe that the semi-classical and mean field residue terms are not of the same order in the combined limit $N^{-1}=\hbar^3$ argument. 

This paper is arranged as follows: we prove the estimate for semi-classical residue in Section \ref{sec:semi-classical_R}, followed by the estimate for mean-field residue in Section \ref{sec:mean_field_residue}. {Then we} conclude the proof of Theorem \ref{thm:main_ben} in Section \ref{sec:proof_mainTheorem}.
In the appendices, for reader's convenience, we list some basic notations and known estimates.

	\vspace{1.5em}	\noindent\textit{Acknowledgment}
	We acknowledge support from the Deutsche Forschungsgemeinschaft through grant CH 955/4-1. 
	Yue Li supported by National Natural Science Foundation of China (12071212).
	Jinyeop Lee received funding from the Deutsche Forschungsgemeinschaft (DFG, German Research Foundation) under Germany’s Excellence Strategy (EXC-2111-390814868).	

\section{Estimate for kinetic residue}\label{sec:kinetic_R}
In this section, we provide the estimate for $\widetilde{\cR}$ in Proposition \ref{prop:tildeR}.
\begin{proof}[Proof of Proposition \ref{prop:tildeR}]
	Note that
	\begin{align*}
		& \left| \int  \ddp \,  \big| \tilde{\cR}{(q,p)} \big| \right| \\
		&
		\leq \hbar   \int \ddp \norm{\nabla_{q} a (f^\hbar_{q,p}) \Psi_{N,t}}\norm{ a (f^\hbar_{q,p}) \Psi_{N,t} }\\
		&
		\leq    \left[ \hbar^2 \int \ddp  \left< \nabla_{q} a (f^\hbar_{q,p}) \Psi_{N,t},   \nabla_{q} a (f^\hbar_{q,p}) \Psi_{N,t} \right> \right]^\frac{1}{2} \left[\int  \ddp\, m_{N,t}(q,p) \right]^\frac{1}{2}\\
		&
		{=}  \left[ \hbar^2 \int \ddp  \left< \nabla_{q} a (f^\hbar_{q,p}) \Psi_{N,t},   \nabla_{q} a (f^\hbar_{q,p}) \Psi_{N,t} \right> \right]^\frac{1}{2} \varrho_{N,t}^\frac{1}{2}\\
		&	=\left[ \hbar^{\frac{1}{2}} \int \ddp  \int  \dw \du\   \nabla_{q} f\left(\frac{w-q}{\sqrt{\hbar}}\right)\nabla_{q} f\left(\frac{u-q}{\sqrt{\hbar}}\right) e^{\frac{\ii}{\hbar}p \cdot (w-u)} \left< \Psi_{N,t},  a^*_w a_u \Psi_{N,t} \right> \right]^\frac{1}{2} \varrho_{N,t}^\frac{1}{2}{(q)}\\
		&	=(2\pi)^3 \left[  \hbar^{\frac{1}{2}+3} \int \dw\ \hbar^{-1}\left|\nabla f\left(\frac{w-q}{\sqrt{\hbar}}\right) \right|^2 \left< \Psi_{N,t},  a^*_w a_w \Psi_{N,t} \right>\right]^\frac{1}{2} \varrho_{N,t}^\frac{1}{2}{(q)}\\
		&\leq \hbar^2 \left[ \int \dw\ \hbar^{-\frac{3}{2}}\left|\nabla f\left(\frac{w-q}{\sqrt{\hbar}}\right) \right|^2 \left< \Psi_{N,t},  a^*_w a_w \Psi_{N,t} \right>\right]^\frac{1}{2} \varrho_{N,t}^\frac{1}{2}{(q)}.
	\end{align*}
	{This implies, by using H\"odler inequality, that}
	\begin{align*}
		&{\Bigg(}\int \dd{q} \left| \int  \ddp\, \big| \tilde{\cR}(q,p)\big| \right|^\frac{5}{4} {\Bigg)^{\frac{4}{5}}}\\
		&\leq {\hbar^2
			\left[\int \dq  \dw\ \hbar^{-\frac{3}{2}}\left|\nabla f\left(\frac{w-q}{\sqrt{\hbar}}\right) \right|^2 \left< \Psi_{N,t},  a^*_w a_w \Psi_{N,t} \right>\right]^\frac{1}{2}
			\; \left[\int dq  \varrho_{N,t}^\frac{5}{3}{(q)}\right]^{\frac{3}{10}}}\\
		& \leq C \hbar^2\|\nabla f\|_2 \left< \Psi_{N,t}, \cN \Psi_{N,t} \right>^{\frac{1}{2}}
		\leq C \hbar^{\frac{1}{2}}.
	\end{align*}
	In the above estimate, we have used the fact that $\|\rho_{N,t}\|_{L^\infty(0,T;L^\frac{5}{3}(\bR^3))}\leq C$, which is a direct result from Appendix \ref{sec:apriori}.
	Moreover, we obatain
	\begin{align*}
		&\bigg|\int\dd{q}\ddp\,\varphi(q)\phi(p) \nabla_{q}\cdot \tilde{\cR}(q,p)\bigg| \\
		&\leq\,{\|\nabla \varphi\|_{L^5(\bR^3)}}{\Bigg(}\int \dd{q} \left| \int  \ddp\, { \phi(p)}\,  \tilde{\cR}(q,p) \right|^\frac{5}{4} {\Bigg)^{\frac{4}{5}}}\\
		&\leq\,{\|\nabla \varphi\|_{L^5(\bR^3)}} { \|\phi\|_{L^\infty(\bR^3)}}\, {\Bigg(}\int \dd{q} \left| \int  \ddp\, \big| \tilde{\cR}(q,p)\big| \right|^\frac{5}{4} {\Bigg)^{\frac{4}{5}}}\\
		%				&\leq {C} {\hbar^2					\left[\int \dq  \dw\ \hbar^{-\frac{3}{2}}\left|\nabla f\left(\frac{w-q}{\sqrt{\hbar}}\right) \right|^2 \left< \Psi_{N,t},  a^*_w a_w \Psi_{N,t} \right>\right]^\frac{1}{2}					\; \left[\int dq  \varrho_{N,t}^\frac{5}{3}\right]^{\frac{3}{10}}}\\
		%				& \leq C \hbar^2\|\nabla f\|_2 \left< \Psi_{N,t}, \cN \Psi_{N,t} \right>^{\frac{1}{2}}
		& \leq C \hbar^{\frac{1}{2}}.
	\end{align*}
	
\end{proof}

\section{Estimate for semi-classical residue} \label{sec:semi-classical_R}

In this section, we will estimate the semi-classical residual term under the assumption with $V\in W^{2,\infty}(\R^3)$ in Proposition \ref{prop:est_R_1}, with which give us the insight to compare the rate between semi-classical and mean-field residuals.

\begin{proof}[Proof of Proposition \ref{prop:est_R_1}]
	First, recall that
	\begin{equation}
		\begin{split}
			\cR_\rms := & \dfrac{1}{(2\pi)^3}  \int \dw_1 \du_1 \dw_2\du_2 \mathrm{d}{{q_2}} \ddp_2    \left(  f^\hbar_{q,p}(w)  \overline{f^\hbar_{q,p}(u)} \right)^{\otimes 2} \\
			&\qquad \left[\int_0^1 \ds\ \nabla V\big(su_1+ (1-s)w_1 - w_2 \big) - \nabla V(q-{{q_2}}) \right]  \gamma_{N,t}^{(2)}(u_1,u_2;w_1,w_2).
		\end{split}
	\end{equation}
	Since $\phi(q)$, $\varphi(p)$ are test functions, we {{see that}}	
	\begin{align*}
		&\bigg|\int \dq \ddp\, \phi(q)\varphi(p)  \nabla_{p} \cdot \cR_\rms(q,p) \bigg| \\
		&
		=  \frac{1}{ (2\pi)^3} \bigg| \int (\dq \ddp)^{\otimes 2}\ \phi(q) \nabla_{p} \varphi(p) \cdot \int \dw_1 \du_1 \dw_2\du_2  \left(  f^\hbar_{q,p}(w)  \overline{f^\hbar_{q,p}(u)} \right)^{\otimes 2} \\
		&
		\qquad \left[\int_0^1 \ds\ \nabla V\big(su_1+ (1-s)w_1 - w_2 \big) - \nabla V(q-{{q_2}}) \right]  \gamma_{N,t}^{(2)}(u_1,u_2;w_1,w_2)\bigg|\\
		&
		= \dfrac{1}{(2\pi\hbar)^3} \bigg| \int (\dq)^{\otimes 2}\ddp \dw_1 \du_1  \dw_2\ \phi(q) \nabla_{p} \varphi(p) f\left(\frac{w_1 -q}{\sqrt{\hbar}}\right) f\left(\frac{u_1 -q}{\sqrt{\hbar}}\right)\\
		&
		\qquad e^{\frac{\ii}{\hbar}p \cdot (w_1-u_1)} \int\du_2 \ddp_2\, e^{\frac{\ii}{\hbar}p_2 \cdot (w_2-u_2)}f\left(\frac{w_2 -{{q_2}}}{\sqrt{\hbar}}\right) f\left(\frac{u_2 -{{q_2}}}{\sqrt{\hbar}}\right) \\
		&\qquad\left[\int_0^1 \ds\ \nabla V\big(su_1+ (1-s)w_1 - w_2 \big) - \nabla V(q-{{q_2}}) \right]
		\gamma_{N,t}^{(2)}(u_1,w_2;w_1,w_2)\bigg|\\
		&
		= \bigg| \int (\dq)^{\otimes 2}\ddp \dw_1 \du_1  \dw_2\ \phi(q) \nabla_{p} \varphi(p) f\left(\frac{w_1 -q}{\sqrt{\hbar}}\right) f\left(\frac{u_1 -q}{\sqrt{\hbar}}\right)\\
		&
		\qquad e^{\frac{\ii}{\hbar}p \cdot (w_1-u_1)} \left|f\left(\frac{w_2 -{{q_2}}}{\sqrt{\hbar}}\right) \right|^2  \left[\int_0^1 \ds\ \nabla V\big(su_1+ (1-s)w_1 - w_2 \big) - \nabla V(q-{{q_2}}) \right] \\
		&
		\qquad \gamma_{N,t}^{(2)}(u_1,w_2;w_1,w_2)\bigg|,
	\end{align*}
	where we applied the fact that $ (2\pi \hbar)^3 \delta_x(y) = \int  e^{\frac{\ii}{\hbar} p \cdot (x-y)} \dd{p}$. Then, inserting $\pm \nabla V(q-w_2)$ and we have
	\begin{align*}
		&
		\leqslant    \bigg| \int (\dq)^{\otimes 2}\ddp \dw_1 \du_1  \dw_2\ \phi(q) \nabla \varphi(p) \cdot f\left(\frac{w_1 -q}{\sqrt{\hbar}}\right) f\left(\frac{u_1 -q}{\sqrt{\hbar}}\right)\\
		&
		\qquad e^{\frac{\ii}{\hbar}p \cdot (w_1-u_1)} \left|f\left(\frac{w_2 -{{q_2}}}{\sqrt{\hbar}}\right) \right|^2  \left[ \int_0^1 \ds\ \nabla V\big(su_1+ (1-s)w_1 - w_2 \big)  - \nabla V(q - w_2) \right] \gamma_{N,t}^{(2)}(u_1,w_2;w_1,w_2)\bigg|\\
		&
		\quad +  \bigg| \int (\dq)^{\otimes 2}\ddp \dw_1 \du_1  \dw_2\ \phi(q) \nabla \varphi(p) \cdot f\left(\frac{w_1 -q}{\sqrt{\hbar}}\right) f\left(\frac{u_1 -q}{\sqrt{\hbar}}\right)\\
		&
		\qquad e^{\frac{\ii}{\hbar}p \cdot (w_1-u_1)} \left|f\left(\frac{w_2 -{{q_2}}}{\sqrt{\hbar}}\right) \right|^2  \left[\nabla V(q - w_2 )  - \nabla V(q - {{q_2}}) \right] \gamma_{N,t}^{(2)}(u_1,w_2;w_1,w_2)\bigg|\\
		& =: I_\rms + J_\rms
	\end{align*}
	where we used integration by part in the second to last equality.\\
	
	Before advancing, recalling \eqref{eq:estimate_oscillation_0}, we split the integral and obtain the following estimate, $\forall\alpha_2\in(\frac{1}{2},1)$,
	\begin{equation}\label{momentum_split}
		\begin{aligned}
			\left|\int \dd{p} \nabla \varphi(p) e^{\frac{i}{\hbar}p\cdot(w-u)}\right|
			&=\left|\int \dd{p}  (\rchi_{(w_1-u_1)\in \Omega_\hbar^{\alpha_2}}+ \rchi_{(w_1-u_1)\in (\Omega_\hbar^{\alpha_2})^c}) \nabla \varphi(p) e^{\frac{i}{\hbar}p\cdot(w-u)}\right|\\
			&\leq \tilde{C} \left(\rchi_{(w_1-u_1)\in \Omega_\hbar^{\alpha_2}} + \hbar^{(1-\alpha_2)s}\right),
		\end{aligned}
	\end{equation}
	where $\tilde{C}$ depends on $\norm{\phi}_{W^{s+1,\infty}}$ and $\supp \phi$.\\
	
	Now we want to estimate the term $I_\rms$ and $J_\rms$ separately. We begin by estimating $I_\rms$,
	\begin{align*}
		I_\rms &
		= \hbar^{\frac{3}{2}} \bigg| \int \dq\ddp \dw_1 \du_1  \dw_2\ \phi(q) \nabla \varphi(p) f\left(\frac{w_1 -q}{\sqrt{\hbar}}\right) f\left(\frac{u_1 -q}{\sqrt{\hbar}}\right)e^{\frac{\ii}{\hbar}p \cdot (w_1-u_1)}\\
		&
		\qquad \left( \int \mathrm{d}\tilde{q}_2\left|f\left(\tilde{q}_2\right) \right|^2 \right) \left[ \int_0^1  \ds\ \nabla V\big(su_1+ (1-s)w_1 - w_2 \big)  - \nabla V(q - w_2) \right]\\
		& \qquad \gamma_{N,t}^{(2)}(u_1,w_2;w_1,w_2)\Big|.
	\end{align*}
	Using $ \norm{D^2 V}_{L^\infty}\leq C$, \eqref{momentum_split}, the definition of $\gamma^{(2)}$ with Cauchy-Schwarz inequality, we have 
	\begin{align*}
		I_\rms %&{\leq}  \norm{D^2 V}_{L^\infty} \tilde{C} \hbar^{\frac{3}{2}} \int \dq\ |\phi(q )|\int \dw_1 \du_1  \dw_2 \left( \rchi_{(w_1-u_1)\in \Omega_\hbar^{\alpha_2}} + \hbar^{(1-\alpha_2)s} \right)   \bigg|f\left(\frac{w_1 -q}{\sqrt{\hbar}}\right) f\left(\frac{u_1 -q}{\sqrt{\hbar}}\right)\bigg|   \\
%		&
%		\qquad \left( |u_1-q|+|w_1-q| \right) \big|  \gamma_{N,t}^{(2)}(u_1,w_2;w_1,w_2)\big|\\
		&
		\leq C \hbar^{\frac{3}{2}} \int \dq\ |\phi(q )|\int \dw_1 \du_1   \left( \rchi_{(w_1-u_1)\in \Omega_\hbar^{\alpha_2}} + \hbar^{(1-\alpha_2)s} \right)   \bigg|f\left(\frac{w_1 -q}{\sqrt{\hbar}}\right) f\left(\frac{u_1 -q}{\sqrt{\hbar}}\right)\bigg|   \\
		&
		\qquad \left( |u_1-q|+|w_1-q| \right) \int \dw_2 \big| \gamma_{N,t}^{(2)}( w_1,w_2; u_1, w_2)\big|\\
%		&
%		=  C  \hbar^{\frac{3}{2}} \int \dq\ |\phi(q )|\int \dw_1 \du_1  \left( \rchi_{(w_1-u_1)\in \Omega_\hbar^{\alpha_2}} + \hbar^{(1-\alpha_2)s} \right)   \bigg|f\left(\frac{w_1 -q}{\sqrt{\hbar}}\right) f\left(\frac{u_1 -q}{\sqrt{\hbar}}\right)\bigg|   \\
%		&
%		\qquad \left( |u_1-q|+|w_1-q| \right) \int \dw_2 \big|  \left<a_{w_2}a_{w_1} \psi_{N,t},  a_{w_2}  a_{u_1} \psi_{N,t} \right>\big|\\
%		&
%		\leq   C  \hbar^{\frac{3}{2}} \int \dq\ |\phi(q )|\int \dw_1 \du_1   \left( \rchi_{(w_1-u_1)\in \Omega_\hbar^{\alpha_2}} + \hbar^{(1-\alpha_2)s} \right)   \bigg|f\left(\frac{w_1 -q}{\sqrt{\hbar}}\right) f\left(\frac{u_1 -q}{\sqrt{\hbar}}\right)\bigg|   \\
%		&
%		\qquad \left( |u_1-q|+|w_1-q| \right) \int \dw_2 \norm{a_{w_2}a_{w_1} \psi_{N,t}} \norm{  a_{w_2}  a_{u_1} \psi_{N,t}}\\
		&
		\leq C  \hbar^{\frac{3}{2}} \int \dq\ |\phi(q )|\int \dw_1 \du_1   \left( \rchi_{(w_1-u_1)\in \Omega_\hbar^{\alpha_2}} + \hbar^{(1-\alpha_2)s} \right)   \bigg|f\left(\frac{w_1 -q}{\sqrt{\hbar}}\right) f\left(\frac{u_1 -q}{\sqrt{\hbar}}\right)\bigg|   \\
		&
		\qquad \left( |u_1-q|+|w_1-q| \right)  \left(\int \dw_2 \norm{ a_{w_2}a_{w_1} \psi_{N,t}}^2\right)^\frac{1}{2} \left(\int \dw_2 \norm{  a_{w_2}  a_{u_1} \psi_{N,t}}^2\right)^\frac{1}{2}\\
		& =: C  \bigg[ i_{\rms,1} + i_{\rms,2}\bigg],
	\end{align*}
	where we use $i_{s,1}$ to be the term with $ \rchi_{(w_1-u_1)\in \Omega_\hbar^\alpha}$, and $i_{s,2}$ to be the other one.
	
	Due to the symmetric property , we can reduce the estimate for $i_{\rms,1}$ into the following 	
	\begin{align}\nonumber
		i_{\rms,1} &  	\leq 2C  \hbar^{\frac{3}{2}} \int \dq\ |\phi(q )|\int \dw_1 \du_1   \rchi_{(w_1-u_1)\in \Omega_\hbar^{\alpha_2}} \bigg|f\left(\frac{w_1 -q}{\sqrt{\hbar}}\right) f\left(\frac{u_1 -q}{\sqrt{\hbar}}\right)\bigg|
		\cdot |u_1-q| \int \dw_2 \norm{ a_{w_2}a_{w_1} \psi_{N,t}}^2\\
%		\nonumber	
		%&  	\leq 2C  \hbar^{\frac{3}{2}} \int \dq\ |\phi(q )|\int \dw_1 \du_1   \rchi_{(w_1-u_1)\in \Omega_\hbar^{\alpha_2}} \bigg|f\left(\frac{w_1 -q}{\sqrt{\hbar}}\right) f\left(\frac{u_1 -q}{\sqrt{\hbar}}\right)\bigg|
		%\cdot |u_1-q|\left<\psi_{N,t}, a^*_{w_1} \cN a_{w_1} \psi_{N,t}\right> \\
		\nonumber	
		&\leq C\hbar^{\frac{3}{2}} \norm{f\left(\frac{u_1 -q}{\sqrt{\hbar}}\right)
			\cdot |u_1-q|}_{L^\infty} |\Omega_\hbar^{\alpha_2}|  \int \dq\ |\phi(q )|\int \dw_1\bigg|f\left(\frac{w_1 -q}{\sqrt{\hbar}}\right)\bigg|\left<\psi_{N,t}, a^*_{w_1} \cN a_{w_1} \psi_{N,t}\right>
		\\
		&\leq C \hbar^{\frac{3}{2}} \hbar^{\frac{1}{2}} \hbar^{3\alpha_2} \hbar^{\frac{3}{2}} \int \dw_1\left<\psi_{N,t}, a^*_{w_1} \cN a_{w_1} \psi_{N,t}\right>
	\leq C \hbar^{3\alpha_2+3+\frac{1}{2}-6}. 	\label{is1}	
	\end{align}
	%Using similar steps with $i_{\rms,1}$, 
	Similarly, noticing that $\phi(q)$ has compact support, one obtains the estimate for the term $i_{\rms,2}$,
	\begin{align*}
		i_{\rms,2} &  	\leq 2C  \hbar^{\frac{3}{2}} \int \dq\ |\phi(q )|\int \dw_1 \du_1   \hbar^{(1-\alpha_2)s} \bigg|f\left(\frac{w_1 -q}{\sqrt{\hbar}}\right) f\left(\frac{u_1 -q}{\sqrt{\hbar}}\right)\bigg|
		\cdot |u_1-q| \int \dw_2 \norm{ a_{w_2}a_{w_1} \psi_{N,t}}^2\\
		&  	\leq 2C  \hbar^{\frac{3}{2}+\frac{3}{2}+\frac{1}{2}} \int \ \mathrm{d}\tilde u \rchi_{|\tilde u|\leq R_1}|f(\tilde u)| |\tilde u|\int \dq\ |\phi(q )|\int \dw_1    \hbar^{(1-\alpha_2)s} \bigg|f\left(\frac{w_1 -q}{\sqrt{\hbar}}\right) \bigg|\int \dw_2 \norm{ a_{w_2}a_{w_1} \psi_{N,t}}^2\\
		&\leq C \hbar^{\frac{3}{2}+\frac{3}{2}+\frac{1}{2}}  \hbar^{(1-\alpha_2)s} \hbar^{\frac{3}{2}}\hbar^{-6}\leq C \hbar^{(1-\alpha_2)s-1}
	\end{align*}
	To balance the order between $i_{\rms,1}$ and $i_{\rms,2}$, the term $s$ is chosen to be
	\[
	s = \left\lceil \frac{3(\alpha_2-\frac{1}{2})}{1-\alpha_2}\right\rceil,
	\]
	where $\alpha_2 \in (\frac{1}{2},1)$. Therefore, we have
	\begin{equation}\label{R1_I}
		I_\rms \leq C \hbar^{\frac{1}{2}+3(\alpha_2-1)}.
	\end{equation}
	
	Now, to estimate $J_\rms$, we recall the estimate in \eqref{momentum_split} and obtain
	\begin{align*}
		J_\rms% &=\bigg| \int (\dq)^{\otimes 2}\ddp \dw_1 \du_1  \dw_2\ \phi(q) \nabla \varphi(p) \cdot f\left(\frac{w_1 -q}{\sqrt{\hbar}}\right) f\left(\frac{u_1 -q}{\sqrt{\hbar}}\right)\\
		%&
		%\qquad e^{\frac{\ii}{\hbar}p \cdot (w_1-u_1)} \left|f\left(\frac{w_2 -{{q_2}}}{\sqrt{\hbar}}\right) \right|^2  \left[\nabla V(q - w_2 )  - \nabla V(q - {{q_2}}) \right] \gamma_{N,t}^{(2)}(u_1,w_2;w_1,w_2)\bigg|\\
		&\leq  C \bigg| \int (\dq)^{\otimes 2}|\phi(q)|\int \dw_1 \du_1 \left( \rchi_{(w_1-u_1)\in \Omega_\hbar^{\alpha_2}} + \hbar^{(1-\alpha_2)s} \right)   f\left(\frac{w_1 -q}{\sqrt{\hbar}}\right) f\left(\frac{u_1 -q}{\sqrt{\hbar}}\right)\\
		&
		\qquad \int \dw_2 \left|f\left(\frac{w_2 -{{q_2}}}{\sqrt{\hbar}}\right) \right|^2  |w_2-{{q_2}}| \gamma_{N,t}^{(2)}(u_1,w_2;w_1,w_2)\bigg|\\
		&
		\leq C \int \dq\ |\phi(q )|\int \dw_1 \du_1   \left( \rchi_{(w_1-u_1)\in \Omega_\hbar^{\alpha_2}} + \hbar^{(1-\alpha_2)s} \right)   \bigg|f\left(\frac{w_1 -q}{\sqrt{\hbar}}\right) f\left(\frac{u_1 -q}{\sqrt{\hbar}}\right)\bigg|   \\
		&
		\qquad \int \mathrm{d}\tilde q \left|f\left(\tilde q\right) \right|^2  h^{\frac{3}{2}} h^{\frac{1}{2}}|\tilde q| \int \dw_2 \big| \gamma_{N,t}^{(2)}( w_1,w_2; u_1, w_2)\big|\\
		&
		\leq   C  \hbar^{2} \int \dq\ |\phi(q )|\int \dw_1 \du_1   \left( \rchi_{(w_1-u_1)\in \Omega_\hbar^{\alpha_2}} + \hbar^{(1-\alpha_2)s} \right)   \bigg|f\left(\frac{w_1 -q}{\sqrt{\hbar}}\right) f\left(\frac{u_1 -q}{\sqrt{\hbar}}\right)\bigg|   \\
		&
		\qquad \left(\int \dw_2 \norm{ a_{w_2}a_{w_1} \psi_{N,t}}^2\right)^\frac{1}{2} \left(\int \dw_2 \norm{  a_{w_2}  a_{u_1} \psi_{N,t}}^2\right)^\frac{1}{2}\\
		& =: C  \bigg[ j_{\rms,1} + j_{\rms,2}\bigg].
	\end{align*}
	The estimate for $j_{\rms,1}$ can be exactly done as in \eqref{is1} for $i_{\rms,1}$, the same for $j_{\rms,2}$ as in $i_{\rms,2}$. Therefore we obtain the same rate for $J_{\rms}$ as in \eqref{R1_I} for $I_{\rms}$. This completes the proof.
\end{proof}

{
	\begin{Remark}
		The key step in the estimates of semi-classical residue is in \eqref{is1}, with which the computational bugs appeared in \cite{Chen2021JSP} can both be fixed by the same technique.
	\end{Remark}
}

\section{Estimate for mean-field residue}\label{sec:mean_field_residue}

In this section, we will estimate the mean-field residue by first showing in Lemma \ref{lem:est_mean-field_osc} that the estimate for mean-field residue term can be reduced to the estimate for the term
\begin{equation}\label{pdbmea}
	\mathds{T}^{(1)} \big|\gamma_{N,t}^{(2)}  - \gamma^{(1)}_{N,t} \otimes \gamma^{(1)}_{N,t}  \big| (u_1;w_1)
\end{equation}
where we denote
\[
\mathds{T}^{(1)} |\gamma^{(2)}  - \gamma^{(1)} \otimes \gamma^{(1)} | (u_1;w_1) := \int \mathrm{d}{y} \left|\gamma^{(2)} (u_1,y;w_1,y) - \gamma^{(1)}(u_1;w_1) \gamma^{(1)}(y;y)\right|.
\]

Then, we prove the estimate for \eqref{pdbmea} in Proposition \ref{prop:mixed_norm} and finally summary the estimation for the mean-field residue in Proposition \ref{prop:est_mean-field_osc}.

\begin{Lemma}\label{lem:est_mean-field_osc}
	Let $\varphi, \phi \in C_0^\infty(\R^3)$. Then, for $\frac{1}{2} < \alpha_1 < 1$ and $s ={\left\lceil \frac{3(\alpha_1 {-} \frac{1}{2})}{2(1-\alpha_1)} \right\rceil}$, we have
	\begin{equation}
		\begin{aligned}
			&\bigg|\int \dq \ddp\, \varphi(q)\phi(p)  \nabla_{p} \cdot \cR_\rmm(q,p) \bigg| \\
			&\quad \leq {C} \norm{\nabla V}_{L^\infty}  \hbar^{3+ \frac{3}{2}(\alpha_1 {-} \frac{1}{2}) + \frac{3}{2}}
			\left(\int \dw_1\du_1\ \left[ \mathds{T}^{(1)} \big|\gamma_{N,t}^{(2)}  - \gamma^{(1)}_{N,t} \otimes \gamma^{(1)}_{N,t}  \big| (u_1;w_1) \right]^2   \right)^\frac{1}{2},
		\end{aligned}
	\end{equation}
	where the constant ${C}$ depends on $\norm{\varphi}_\infty$, $\norm{\nabla \phi}_{W^{s,\infty}}$, $\supp\phi$, $\norm{f}_{L^\infty \cap H^1}$, $\supp f$.
\end{Lemma}
\begin{proof}
	Recall that, in \eqref{eq:bbgky_remainder_1_ben}, we defined the mean-field residue such that
	\begin{equation}
		\begin{aligned}
			{\cR_\rmm} &:= 	\frac{1}{(2\pi)^3}  \int \dw_1 \du_1 \dw_2\du_2 \dq_2 \ddp_2\,    \left(  f^\hbar_{q,p}(w)  \overline{f^\hbar_{q,p}(u)} \right)^{\otimes 2}\\
			&
			\qquad  \nabla V(q-q_2) \bigg[ \gamma_{N,t}^{(2)}(u_1,u_2;w_1,w_2) -\gamma_{N,t}^{(1)}(u_1;w_1)\gamma_{N,t}^{(1)}(u_2;w_2)\bigg].
		\end{aligned}
	\end{equation}
	Then one obtains
	\begin{align*}
		&\bigg|\int \dq \ddp\, \varphi(q) \phi(p) \nabla_{p} \cdot \cR_\rmm (q,p)\bigg| \\
		&
		= \frac{1}{(2\pi)^3}\bigg| \int (\dq \ddp)^{\otimes 2} (\dw \du)^{\otimes 2}\, \varphi(q) \nabla\phi(p)  \cdot  \left(  f^\hbar_{q,p}(w)  \overline{f^\hbar_{q,p}(u)} \right)^{\otimes 2} \nabla V(q-q_2) \\
		&
		\qquad \left[\gamma_{N,t}^{(2)}(u_1,u_2;w_1,w_2)  - \gamma^{(1)}_{N,t}(u_1; w_1) \gamma^{(1)}_{N,t}(u_2; w_2)\right] \bigg|\\
		%		&
		%		= \frac{1}{(2\pi\hbar)^3}\bigg| \int (\dq \ddp)^{\otimes 2} (\dw \du)^{\otimes 2}\,  \varphi(q) \nabla\phi(p)  \cdot \left(f\left(\frac{w-q}{\sqrt{\hbar}}\right) f\left(\frac{u-q}{\sqrt{\hbar}}\right) e^{\frac{\ii}{\hbar} p\cdot(w-u)}\right)^{\otimes 2}  \nabla V(q-q_2) \\
		%		&
		%		\qquad \left[\gamma_{N,t}^{(2)}(u_1,u_2;w_1,w_2)  - \gamma^{(1)}_{N,t}(u_1; w_1) \gamma^{(1)}_{N,t}(u_2; w_2)\right] \bigg|\\
		&
		= \frac{1}{(2\pi\hbar)^3}\bigg| \int (\dq \dw \du)^{\otimes 2}\, \left(f\left(\frac{w-q}{\sqrt{\hbar}}\right) f\left(\frac{u-q}{\sqrt{\hbar}}\right) \right)^{\otimes 2}\left(\int \ddp\, \varphi(q) \nabla\phi(p)  \cdot e^{\frac{\ii}{\hbar} p\cdot(w_1-u_1)}\right) \\
		&
		\qquad \nabla V(q-q_2) \left(\int \ddp_2\, e^{\frac{\ii}{\hbar} p_2\cdot(w_2-u_2)} \right) \left[\gamma_{N,t}^{(2)}(u_1,u_2;w_1,w_2)  -  \gamma^{(1)}_{N,t}(u_1; w_1) \gamma^{(1)}_{N,t}(u_2; w_2) \right] \bigg|\\
		&
		= \bigg| \int (\dq)^{\otimes 2} \dw_1\du_1\dw_2 \,   f\left(\frac{w_1-q}{\sqrt{\hbar}}\right) f\left(\frac{u_1-q}{\sqrt{\hbar}}\right) \left|f\left(\frac{w_2-q_2}{\sqrt{\hbar}}\right)\right|^2  \left(\int \ddp\, \varphi(q) \nabla\phi(p)  \cdot e^{\frac{\ii}{\hbar} p_1\cdot(w_1-u_1)}\right) \\
		&
		\qquad \nabla V(q-q_2) \left[\gamma_{N,t}^{(2)}(u_1,w_2;w_1,w_2)  -  \gamma^{(1)}_{N,t}(u_1; w_1) \gamma^{(1)}_{N,t}(w_2; w_2)\right] \bigg|,
	\end{align*}
	where we use the weighted Dirac-Delta function in the last equality, i.e.,
	\begin{equation}\label{eq:diract_delta}
		\frac{1}{(2\pi \hbar)^3}\int \ddp_2\ e^{\frac{\ii}{\hbar} p_2\cdot(w_2-u_2)} = \delta_{w_2}(u_2).
	\end{equation}
	Now, splitting the domains of $w_1$ and $u_1$ into two, namely, with the characteristic functions $\rchi_{(w_1 - u_1)\in \Omega_\hbar^{\alpha_1}}$ and $\rchi_{(w_1 - u_1)\in(\Omega_\hbar^{\alpha_1})^c}$ whose domain $\Omega_\hbar^{\alpha_1}$ is defined in \eqref{eq:estimate_oscillation_omega}, we have
	
	\begin{align*}
		&\leqslant \bigg| \int (\dq)^{\otimes 2}  \varphi(q) \int \dw_1\du_1\dw_2\   f\left(\frac{w_1-q}{\sqrt{\hbar}}\right) f\left(\frac{u_1-q}{\sqrt{\hbar}}\right) \left|f\left(\frac{w_2-q_2}{\sqrt{\hbar}}\right)\right|^2  \\
		&
		\qquad \left(\int \ddp\, \rchi_{(w_1 - u_1)\in \Omega_\hbar^{\alpha_1}} e^{\frac{\ii}{\hbar} p\cdot(w_1-u_1)} \nabla\phi(p)   \right) \cdot \nabla V(q-q_2) \left[\gamma_{N,t}^{(2)}(u_1,w_2;w_1,w_2)  - \gamma^{(1)}_{N,t}(u_1; w_1) \gamma^{(1)}_{N,t}(w_2; w_2)\right] \bigg|\\
		&
		+  \bigg| 	\int (\dq)^{\otimes 2}\varphi(q) \int \dw_1\du_1\dw_2\    f\left(\frac{w_1-q}{\sqrt{\hbar}}\right) f\left(\frac{u_1-q}{\sqrt{\hbar}}\right) \left|f\left(\frac{w_2-q_2}{\sqrt{\hbar}}\right)\right|^2  \\
		&
		\qquad \left(\int \ddp\, \rchi_{(w_1 - u_1)\in (\Omega_\hbar^{\alpha_1})^c}  e^{\frac{\ii}{\hbar} p\cdot(w_1-u_1)} \nabla \phi(p)  \right) \cdot  \nabla V(q-q_2) \left[\gamma_{N,t}^{(2)}(u_1,w_2;w_1,w_2)  - \gamma^{(1)}_{N,t}(u_1; w_1) \gamma^{(1)}_{N,t}(w_2; w_2)\right] \bigg|\\
		& =: \text{I}_{\rmm} + \text{J}_{\rmm}.
	\end{align*}
	First, considering the term $\text{J}_{\rmm}$, by the change of variable $\sqrt{\hbar}\,\tilde{q}_2 = w_2 - q_2$, we obtain
	\begin{align*}
		\text{J}_{\rmm}& =  \bigg|	\int \dq\, \varphi(q) \int \dw_1\du_1\dw_2\, f\left(\frac{w_1-q}{\sqrt{\hbar}}\right) f \left(\frac{u_1-q}{\sqrt{\hbar}}  \right) \left(\hbar^{\frac{3}{2}} \int  \mathrm{d}\tilde{q}_2\   |f(\tilde{q}_2)|^2  \nabla V(q-w_2 + \sqrt{\hbar}\,\tilde{q}_2)\right) \\
		&
		\qquad \left(\int \ddp\, \rchi_{(w_1 - u_1)\in (\Omega_\hbar^{\alpha_1})^c} \nabla  \phi(p) e^{\frac{\ii}{\hbar} p\cdot(w_1-u_1)}  \right)\left(\gamma_{N,t}^{(2)}(u_1,w_2;w_1,w_2)  - \gamma^{(1)}_{N,t}(u_1; w_1) \gamma^{(1)}_{N,t}(w_2; w_2)\right) \bigg|\\
		%		&
		%		\leq C \norm{\nabla V}_{L^\infty}  \hbar^{\frac{3}{2}} 	\int \dq\,  |\varphi(q)| \int \dw_1\du_1\dw_2\  \left|f\left(\frac{w_1-q}{\sqrt{\hbar}}\right) f\left(\frac{u_1-q}{\sqrt{\hbar}}\right) \right|\\
		%		&
		%		\qquad \left|\int \ddp\, \rchi_{(w_1 - u_1)\in (\Omega_\hbar^{\alpha_1})^c} \nabla  \phi(p) e^{\frac{\ii}{\hbar} p\cdot(w_1-u_1)}\right| \left|  \gamma_{N,t}^{(2)}(u_1,w_2;w_1,w_2) -  \gamma_{N,t}(u_1; w_1) \gamma_{N,t}(w_2; w_2) \right|\\
		&
		\leq C \norm{\nabla V}_{L^\infty}  \hbar^{\frac{3}{2}}   \int \dq\, | \varphi(q)| \int \dw_1\du_1\,  \bigg| f\left(\frac{w_1-q}{\sqrt{\hbar}}\right) f\left(\frac{u_1-q}{\sqrt{\hbar}}\right)\bigg|\rchi_{|w_1-u_1|\leq 2R_1 \sqrt{\hbar}}\\
		&
		\qquad \left|\int \ddp\, \rchi_{(w_1 - u_1)\in (\Omega_\hbar^{\alpha_1})^c} \nabla  \phi(p) e^{\frac{\ii}{\hbar} p\cdot(w_1-u_1)}\right|\,  \int\dw_2 \left|\gamma_{N,t}^{(2)}(u_1,w_2;w_1,w_2)  - \gamma^{(1)}_{N,t}(u_1; w_1) \gamma^{(1)}_{N,t}(w_2; w_2)\right|  .
	\end{align*}
	where we have used that supp$f\subset B_{R_1}$.
	Recall again from Lemma \ref{lem:estimate_oscillation} that we have
	\begin{equation}
		\left|\int \ddp\, \rchi_{(w_1 - u_1)\in (\Omega_\hbar^{\alpha_1})^c} e^{\frac{\ii}{\hbar} p\cdot(w_1-u_1)} \nabla  \phi(p) \right| \leq \norm{\nabla \phi}_{W^{s,\infty}}  \hbar^{(1-\alpha_1)s},
	\end{equation}
	for $s$ to be chosen later.
	%	Hence,
	%	\begin{align*}
		%		\text{J}_{\rmm} & \leq C \norm{\nabla \phi}_{W^{s,\infty}} \norm{\nabla V}_{L^\infty}  \hbar^{\frac{3}{2}+ (1-\alpha_1)s}  \int \dq\, |\varphi(q)|\int \dw_1\du_1\,  \bigg|  f\left(\frac{w_1-q}{\sqrt{\hbar}}\right) f\left(\frac{u_1-q}{\sqrt{\hbar}}\right)\bigg| \\
		%		&
		%		\qquad  \mathds{T}^{(1)} \left|\gamma_{N,t}^{(2)}  - \gamma^{(1)}_{N,t} \otimes \gamma^{(1)}_{N,t} \right|(u_1;w_1)   \rchi_{|w_1-q|\leq R_1\sqrt{\hbar}}\;\rchi_{|u_1-q|\leq R_1 \sqrt{\hbar}},
		%	\end{align*}
	Hence, together with H\"older's inequality we get
	\begin{align*}
		\text{J}_{\rmm} & \leq C  \norm{\nabla \phi}_{W^{s,\infty}} \norm{\nabla V}_{L^\infty} \hbar^{\frac{3}{2}+ (1-\alpha_1)s} \int \dq\, |\varphi(q)| \left(\int \dw_1\du_1\, \rchi_{|w_1-u_1|\leq 2R_1 \sqrt{\hbar}} \left|f\left(\frac{w_1-q}{\sqrt{\hbar}}\right) f\left(\frac{u_1-q}{\sqrt{\hbar}}\right)\right|^2 \right)^\frac{1}{2}\\
		&
		\qquad \left(\int \dw_1\du_1\, \left[ \mathds{T}^{(1)} \big|\gamma_{N,t}^{(2)}  -  \gamma^{(1)}_{N,t} \otimes \gamma^{(1)}_{N,t} \big|(u_1;w_1) \right]^2 \rchi_{|w_1-q|\leq R_1\sqrt{\hbar}}\right)^\frac{1}{2}\\
		&
		\leq  C \norm{\varphi}_{L^\infty} \norm{\nabla \phi}_{W^{s,\infty}} \norm{\nabla V}_{L^\infty} \hbar^{\frac{3}{2}+ (1-\alpha_1)s}
		\left(\hbar^3 \int \mathrm{d}\tilde{w}_1\mathrm{d}\tilde{u}_1\, \rchi_{|\tilde{w}_1-\tilde{u}_1|\leq 2R_1 } |f\left(\tilde{w}\right) f\left(\tilde{u}\right)|^2 \right)^\frac{1}{2}\\
		&
		\qquad \int \dq  {\Bigg(}\int \dw_1\du_1\, \left[ \mathds{T}^{(1)} \big|\gamma_{N,t}^{(2)}  -  \gamma^{(1)}_{N,t} \otimes \gamma^{(1)}_{N,t} \big|(u_1;w_1) \right]^2  \rchi_{|w_1-q|\leq R_1 \sqrt{\hbar}} \Bigg)^\frac{1}{2}\\
		&
		\leq C   \norm{\varphi}_{L^\infty} \norm{\nabla \phi}_{W^{s,\infty}} \norm{\nabla V}_{L^\infty}  \hbar^{ 3 + (1-\alpha_1)s } \left(\int \mathrm{d}\tilde{w}_1\mathrm{d}\tilde{u}_1\, |f\left(\tilde{w}\right) f\left(\tilde{u}\right)|^2\right)^\frac{1}{2} \\
		&
		\qquad \hbar^{\frac{3}{2}} \int \mathrm{d}\tilde{q}_1 \rchi_{|\tilde{q}_1|\leq R_1 }   \left(\int \dw_1\du_1\,\left[ \mathds{T}^{(1)} \big|\gamma_{N,t}^{(2)}  - \gamma^{(1)}_{N,t} \otimes \gamma^{(1)}_{N,t} \big|(u_1;w_1) \right]^2  \right)^\frac{1}{2}\\
		&
		\leq C  \norm{\varphi}_{L^\infty} \norm{\nabla \phi}_{W^{s,\infty}} \norm{\nabla V}_{L^\infty} \hbar^{ 3 + (1-\alpha_1)s+\frac{3}{2}} \left(\int \dw_1\du_1\, \left[ \mathds{T}^{(1)} \big |\gamma_{N,t}^{(2)}  -  \gamma^{(1)}_{N,t} \otimes \gamma^{(1)}_{N,t} \big|(u_1;w_1) \right]^2  \right)^\frac{1}{2}.
	\end{align*}
	Now, we focus on $\text{I}_{m}$.
	Using the fact that
	$
	\left|\int \ddp\,  e^{\frac{\ii}{\hbar} p\cdot(w_1-u_1)} \nabla  \phi(p)\right|\leq \norm{\nabla \phi}_{L^1},
	$
	we obtain the following estimate:
	\begin{align*}
		\text{I}_{\rmm}& \leq C \norm{\nabla \phi}_{L^1}\norm{\nabla V}_{L^\infty} \int \dq\, |\varphi(q)|   \left( \int \dw_1\du_1\, \rchi_{|w_1-u_1|\leq \hbar^{\alpha_1}}  \rchi_{|w_1-u_1|\leq 2R_1 \sqrt{\hbar}} \left|f\left(\frac{w_1-q}{\sqrt{\hbar}}\right) f\left(\frac{u_1-q}{\sqrt{\hbar}}\right)\right|^2 \right)^\frac{1}{2}\\
		&
		\qquad \hbar^{\frac{3}{2}} \int \mathrm{d}{\tilde{q}_2} |f(\tilde{q}_2)|^2 \left(\int \dw_1\du_1\, \left[ \mathds{T}^{(1)} \big|\gamma_{N,t}^{(2)}  - \gamma^{(1)}_{N,t} \otimes \gamma^{(1)}_{N,t} \big| (u_1;w_1) \right]^2   \rchi_{|w_1-q|\leq R_1 \sqrt{\hbar}}  \right)^\frac{1}{2}\\
		&
		\leq C \norm{\varphi}_{L^\infty} \norm{\nabla \phi}_{L^1} \norm{\nabla V}_{L^\infty}  \hbar^{\frac{3}{2}} \left(\hbar^3  \int \mathrm{d}\tilde{w}_1\mathrm{d}\tilde{u}_1\, \rchi_{|\tilde{w}_1-\tilde{u}_1|\leq \hbar^{\alpha_1 {-} \frac{1}{2}}}  \rchi_{|\tilde{w}_1-\tilde{u}_1|\leq 2R_1} \left|f(\tilde{w}_1) f(\tilde{u}_1) \right|^2 \right)^\frac{1}{2}\\
		&
		\qquad \int\mathrm{d}{q} \left(\int \dw_1\du_1\, \left[ \mathds{T}^{(1)} \big|\gamma_{N,t}^{(2)}  - \gamma^{(1)}_{N,t} \otimes \gamma^{(1)}_{N,t} \big| (u_1;w_1) \right]^2  \rchi_{|w_1-q|\leq R_1 \sqrt{\hbar}}   \right)^\frac{1}{2}\\
		&
		\leq C \norm{\varphi}_{L^\infty} \norm{\nabla \phi}_{L^1}\norm{\nabla V}_{L^\infty} \hbar^3 \left( \int \mathrm{d}\tilde{w}_1\mathrm{d}\tilde{u}_1\, \rchi_{|\tilde{w}_1-\tilde{u}_1|\leq \hbar^{\alpha_1 {-} \frac{1}{2}}}  \left|f(\tilde{w}_1) f(\tilde{u}_1) \right|^2 \right)^\frac{1}{2}\\
		&
		\qquad \hbar^{ \frac{3}{2}}\int \mathrm{d}{\tilde{q}_1} \rchi_{|\tilde{q}_1|\leq R_1}
		\left(\int \dw_1\du_1\, \left[ \mathds{T}^{(1)} \big|\gamma_{N,t}^{(2)}  - \gamma^{(1)}_{N,t} \otimes \gamma^{(1)}_{N,t}  \big| (u_1;w_1) \right]^2  \right)^\frac{1}{2},
	\end{align*}
	which together with
	\begin{align*}
		&\int \mathrm{d}\widetilde{w}_1\, |f(\widetilde{w}_1)|^2 \int \mathrm{d}\widetilde{u}_1 \rchi_{|\widetilde{w}_1-\widetilde{u}_1| \leq \hbar^{\alpha_1-\frac{1}{2}}}  |f(\widetilde{u}_1)|^2\leq \norm{f}_{L^\infty(\bR^3)}^2 \norm{f}_{L^2(\bR^3)}^2 \hbar^{3(\alpha_1-\frac{1}{2})},
	\end{align*}
	implies immediately that
	\begin{equation*}
		\text{I}_{\rmm} \leq C \norm{\varphi}_{L^\infty} \norm{\nabla \phi}_{L^1} \norm{\nabla V}_{L^\infty}  \hbar^{3+{\frac{3}{2}}(\alpha_1 {-} \frac{1}{2}) + \frac{3}{2}} \left(\int \dw_1\du_1\ \left[ \mathds{T}^{(1)} \big|\gamma_{N,t}^{(2)}  - \gamma^{(1)}_{N,t} \otimes \gamma^{(1)}_{N,t}  \big| (u_1;w_1) \right]^2  \right)^\frac{1}{2}.
	\end{equation*}
	To balance the order between $\text{I}_{\rmm}$ and $\text{J}_{\rmm}$, $s$ is chosen to be
	\[
	s = {\left\lceil \frac{3(\alpha_1 {-} \frac{1}{2})}{2(1-\alpha_1)} \right\rceil},
	\]
	for $\alpha_1 \in \left[0,1\right)$. Therefore, we obtained the desired result:
	\begin{equation*}
		\begin{aligned}
			\bigg|&\int \dq \ddp\, \ \varphi(q) \phi(p) \nabla_{p} \cdot \cR_\rmm \bigg| \\
			&
			\leq C  \norm{\varphi}_{L^\infty} \norm{\nabla \phi}_{W^{s,\infty}} \norm{\nabla V}_{L^\infty}  \hbar^{3+ {\frac{3}{2}}(\alpha_1 {-} \frac{1}{2}) + \frac{3}{2}} \left(\int \dw_1\du_1\ \left[ \mathds{T}^{(1)} \big|\gamma_{N,t}^{(2)}  - \gamma^{(1)}_{N,t} \otimes \gamma^{(1)}_{N,t}  \big| (u_1;w_1) \right]^2   \right)^\frac{1}{2}.
		\end{aligned}
	\end{equation*}
\end{proof}

Next, we want to bound the term with the `mixed'-norm, i.e.,
\begin{equation}\label{eq_mixnorm_0}
	\int \dw_1\du_1\ \left[ \mathds{T}^{(1)} \big|\gamma_{N,t}^{(2)}  - \gamma^{(1)}_{N,t} \otimes \gamma^{(1)}_{N,t}  \big| (u_1;w_1) \right]^2.
\end{equation}
The following proposition provides the estimate of \eqref{eq_mixnorm_0}:

\begin{Proposition}\label{prop:mixed_norm}
	Let $\gamma^{(k)}_{N,t}$ be $k$-particle reduced density matrix associated with $\Psi_{N,t}$, $\omega_{N,t}$ be the solution of the Hartree-Fock equation in \eqref{eq:hartree-fock}. Suppose the assumption for Theorem \ref{thm:main_ben} holds. Then the following inequalities hold for all $t\in \R$:
	
	\begin{equation}\label{eq:theorem_HS}
		\norm{\gamma^{(1)}_{N,t}- \omega_{N,t}}_{\mathrm{HS}} \leq C.
	\end{equation}
	and
	\begin{equation}\label{eq:theorem_trace}
		\trnorm{\gamma^{(1)}_{N,t}- \omega_{N,t} }\leq C \sqrt{N}.
	\end{equation}
	Furthermore, it holds that
	\begin{equation}\label{eq:mixed_norm_0}
		\left(\int \dw_1\du_1\ \left[ \mathds{T}^{(1)} \big|\gamma_{N,t}^{(2)}  - \omega_{N,t} \otimes \omega_{N,t}  \big| \right]^2(u_1;w_1)   \right)^\frac{1}{2} \leq   C\, N
	\end{equation}
{where $C$ depends on $t$ but is independent of $N$.}
\end{Proposition}
\label{sec:mixednorm}

The proof of Proposition \ref{prop:mixed_norm} requires the following results from \cite{benedikter2014mean}, namely:

\begin{Lemma}[Lemma 3.1 of \cite{benedikter2014mean}]\label{lem:bound_O}
	Let $\dd{\Gamma}(O)$ be the second quantization of any bounded operator $O$ on $L^2(\R^3)$, i.e.
	\[
	\dd{\Gamma}(O) :=  \int \dx\dy\ O(x;y) a^*_x a_y.
	\]
	For any $\Psi \in \mathcal{F}_a$, the following inequalities hold
	\begin{equation}\label{eq:bound_O_1}
		\norm{\dd{\Gamma(O)\Psi}} \leq \norm{O} \norm{\cN \Psi }.
	\end{equation}
	If furthermore $O$ is a Hilbert-Schmidt operator, we have the following bounds:
	\begin{align}
		\norm{\dd{\Gamma (O) \Psi}} &\leq \norm{O}_{\mathrm{HS}} \norm{\cN^{1/2} \Psi },\label{eq:bound_O_2} \\
		\norm{\int \mathrm{d}{x} \mathrm{d}{y} O(x;y) a_x a_y \Psi } & \leq  \norm{O}_{\mathrm{HS}} \norm{\cN^{1/2} \Psi },\label{eq:bound_O_3} \\
		\norm{\int \mathrm{d}{x} \mathrm{d}{y} O(x;y) a_x^* a_y^* \Psi } & \leq 2 \norm{O}_{\mathrm{HS}}\norm{(\cN+1)^{1/2} \Psi }.\label{eq:bound_O_4}
	\end{align}
	Finally, if $O$ is a trace class operator, we obtain
	\begin{align}
		\norm{\dd{\Gamma (O) \Psi}} &\leq 2  \norm{O}_{\Tr},\label{eq:bound_O_5} \\
		\norm{\int \mathrm{d}{x} \mathrm{d}{y} O(x;y) a_x a_y \Psi } &\leq 2  \norm{O}_{\Tr},\label{eq:bound_O_6} \\
		\norm{\int \mathrm{d}{x} \mathrm{d}{y} O(x;y) a_x^* a_y^* \Psi } & \leq 2  \norm{O}_{\Tr}\label{eq:bound_O_7},
	\end{align}
	where $\norm{O}_{\Tr} := \Tr |O| = \Tr \sqrt{O^*O}$.
\end{Lemma}

\begin{Lemma}[Proposition 3.4 of \cite{benedikter2014mean}]\label{lem:propagation}
	Suppose the assumption for Theorem \ref{thm:main_ben} holds. Then, there exist constants $K,c>0$ depending only on potential $V$ such that
	\begin{align*}
		\sup_{p\in\bR^{3}}\frac{1}{1+|p|}\Tr|[\omega_{N,t},e^{\ii p\cdot x}]| & \leq K N\hbar\, C(t)\\
		\Tr|[\omega_{N,t},\hbar\nabla]| & \leq KN\hbar \ C(t).
	\end{align*}
\end{Lemma}

\begin{Lemma}[Theorem 3.2 of \cite{benedikter2014mean}]\label{lem:est_U}
	Let $\cU_{N}(t;s)$ be the quantum fluctuation dynamics defined in \eqref{eq:flunctuation} and $\cN$ be the number operator. If the assumptions in Lemma \ref{lem:propagation} hold. Then for $\xi_N\in\mathcal{F}_a$ with $\langle \xi_N , \cN^k \xi_N \rangle \leq C$ for any $k \geq 1$, we have the following inequality:
	\begin{equation}
		\norm{(\cN+1)^k \, \cU_{N}(t;0)\xi_N} \leq   C{(k,t)} .
	\end{equation}
\end{Lemma}

\begin{Remark}
	Here in this paper we only need the result for initial data {$\xi_N=\Omega$ where $\Omega$ is the vacuum state give in Appendix \ref{sec:FockSpace}.}
\end{Remark}

Now, we are ready to provide the proof of Proposition \ref{prop:mixed_norm}.

\begin{proof}[Proof of Proposition \ref{prop:mixed_norm}]
	
	The proof of the inequalities \eqref{eq:theorem_HS} and \eqref{eq:theorem_trace} follows by modifying Theorem 2.1 of \cite{benedikter2014mean}. In particular, from equation (4.3) in \cite{benedikter2014mean}, we obtain
	\begin{align*}
		\hsnorm{\gamma^{(1)}_{N,t} - \omega_{N,t}}  &\leq C  \norm{\cN^\frac{1}{2} \cU_N(t;0) \xi_N},\\
		\trnorm{ \gamma^{(1)}_{N,t} - \omega_{N,t} } &\leq  C  \sqrt{N} \norm{ \cN  \cU_N(t;0) \xi_N},
	\end{align*}
	by choosing the appropriate operator $O$ as discussed in \cite{benedikter2014mean}. Our results for \eqref{eq:theorem_HS} and \eqref{eq:theorem_trace} are obtained by applying Lemma \ref{lem:est_U} and taking the assumption that  $\norm{(\cN+1)\xi_N}\leq C$.
	
	Therefore, it remains to prove for \eqref{eq:mixed_norm_0}. As remarked previously, the trace norm and Hilbert-Schmidt norm of the difference between $\gamma_{N,t}^{(k)}$ and $\omega_{N,t}^{(k)}$ are obtained separately with the help of Wick's theorem for $k \geq 2$ in \cite{benedikter2014mean}. For our term, however, we do not directly use Wick's theorem to compute \eqref{eq:mixed_norm_0} as each terms requires similar but still unique method when taking the estimation.
	
	Simplifying the notation $\cR_{t} := \cR_{\mathcal{V}_{N,t}}$, where $\cR_{\mathcal{V}_{N,t}}$ is the Bogoliubov transformation given in \eqref{eq:exciting_anil_creation}, we have, from the definition of a $2$-particle reduced density matrix and \eqref{eq:exciting_anil_creation}. that
	\begin{align}
		&\gamma^{(2)}_{N,t}(x_1,x_2;y_1,y_2) \nonumber\\
		& = \left< \xi_N, \cU^*_N(t;0) \cR^*_t a^*_{y_1} a^*_{y_2} a_{x_2} a_{x_1} \cR_t    \cU_N(t;0) \xi_N \right>\nonumber\\
		&
		= \left< \xi_N, \cU^*_N(t;0) \ \cR^*_t a^*_{y_1} \cR_t \cR^*_t a^*_{y_2} \cR_t \cR^*_t a_{x_2} \cR_t \cR^*_t a_{x_1} \cR_t  \cU_N(t;0) \xi_N \right>\nonumber\\
		&
		= \bigg<  \xi_N, \cU^*_N(t;0)  \left( a^*(\rmu_{t,{y_1}}) + a(\bar{\rmv}_{t,{y_1}}) \right) \left( a^*(\rmu_{t,{y_2}}) + a(\bar{\rmv}_{t,{y_2}}) \right) \nonumber\\
		&\qquad \left(a(\rmu_{t,x_2}) + a^*(\bar{\rmv}_{t,x_2})\right) \left(a(\rmu_{t,x_1}) + a^*(\bar{\rmv}_{t,x_1})\right)    \cU_N(t;0)  \xi_N  \big>\nonumber\\
		&
		= \bigg<  \xi_N, \cU^*_N(t;0) \bigg[ a(\bar{\rmv}_{t,y_1})a(\bar{\rmv}_{t,y_2})a(\rmu_{t,x_2})a(\rmu_{t,x_1}) + a(\bar{\rmv}_{t,y_1})a(\bar{\rmv}_{t,y_2})a^*(\bar{\rmv}_{t,x_2})a(\rmu_{t,x_1}) \nonumber\\
		&
		\qquad + a(\bar{\rmv}_{t,y_1})a(\bar{\rmv}_{t,y_2})a(\rmu_{t,x_2})a^*(\bar{\rmv}_{t,x_1}) + {  a(\bar{\rmv}_{t,y_1})a(\bar{\rmv}_{t,y_2})a^*(\bar{\rmv}_{t,x_2})a^*(\bar{\rmv}_{t,x_1})}\nonumber\\
		&
		\qquad + a(\bar{\rmv}_{t,y_1})a^*(\rmu_{t,y_2})a(\rmu_{t,x_2})a(\rmu_{t,x_1}) + a(\bar{\rmv}_{t,y_1})a^*(\rmu_{t,y_2})a^*(\bar{\rmv}_{t,x_2})a(\rmu_{t,x_1})\nonumber\\
		&
		\qquad + a(\bar{\rmv}_{t,y_1})a^*(\rmu_{t,y_2})a(\rmu_{t,x_2})a^*(\bar{\rmv}_{t,x_1}) + a(\bar{\rmv}_{t,y_1})a^*(\rmu_{t,y_2})a^*(\bar{\rmv}_{t,x_2})a^*(\bar{\rmv}_{t,x_1})\nonumber\\
		&
		\qquad + a^*(\rmu_{t,y_1})a(\bar{\rmv}_{t,y_2})a(\rmu_{t,x_2})a(\rmu_{t,x_1}) + a^*(\rmu_{t,y_1})a(\bar{\rmv}_{t,y_2})a^*(\bar{\rmv}_{t,x_2})a(\rmu_{t,x_1})\nonumber\\
		&
		\qquad + a^*(\rmu_{t,y_1})a(\bar{\rmv}_{t,y_2})a(\rmu_{t,x_2})a^*(\bar{\rmv}_{t,x_1}) + a^*(\rmu_{t,y_1})a(\bar{\rmv}_{t,y_2})a^*(\bar{\rmv}_{t,x_2})a^*(\bar{\rmv}_{t,x_1})\nonumber\\
		&
		\qquad + a^*(\rmu_{t,y_1})a^*(\rmu_{t,y_2})a(\rmu_{t,x_2}) a(\rmu_{t,x_1})+ a^*(\rmu_{t,y_1})a^*(\rmu_{t,y_2})a^*(\bar{\rmv}_{t,x_2})a(\rmu_{t,x_1})\nonumber\\
		&
		\qquad + a^*(\rmu_{t,y_1})a^*(\rmu_{t,y_2}) a(\rmu_{t,x_2})a^*(\bar{\rmv}_{t,x_1})+ a^*(\rmu_{t,y_1})a^*(\rmu_{t,y_2})a^*(\bar{\rmv}_{t,x_2})a^*(\bar{\rmv}_{t,x_1})\bigg] \cU_N(t;0)  \xi_N  \bigg>,\label{eq:gamma_omega_0}
	\end{align}
	where we use \eqref{eq:exciting_anil_creation} in the third equality. 
	Therefore, we obtain,	using the fact that $ \left<\bar{\rmv}_{t,x}, {\rmu}_{t,y} \right> = 0$, $\left<\rmu_{t,x}, \bar{\rmv}_{t,y} \right>  = 0$,  $\left<\bar{\rmv}_{t,x}, \bar{\rmv}_{t,y}\right> = {\omega_{N,t}(y;x)}$, and CAR,
	\begin{align}%\label{eq:rearranged_CAR_1}
		&\int \dx_1 \dx_2 \mathrm{d}z_1 \mathrm{d}z_2\ O_1(x_1;z_1) O_2(x_2;z_2)\left( \gamma^{(2)}_{N,t}(z_1,z_2;x_1,x_2) - \omega_{N,t} (z_1; x_1)\omega_{N,t} (z_2; x_2)  \right)\nonumber\\
		&
		= \int \dx_1 \dx_2 \mathrm{d}z_1 \mathrm{d}z_2\ O_1(x_1;z_1) O_2(x_2;z_2)  \bigg<  \xi_N, \cU^*_N(t;0) \bigg[ a(\bar{\rmv}_{t,x_1}) a(\rmu_{t,z_1}) a(\bar{\rmv}_{t,x_2}) a(\rmu_{t,z_2})  \nonumber\\
		&
		\qquad + a(\bar{\rmv}_{t,x_1}) a(\rmu_{t,z_1}) a^*(\rmu_{t,x_2}) a(\rmu_{t,z_2}) + a(\bar{\rmv}_{t,x_1}) a(\rmu_{t,z_1}) a^*(\rmu_{t,x_2})a^*(\bar{\rmv}_{t,z_2}) + a^*(\bar{\rmv}_{t,z_1}) a(\bar{\rmv}_{t,x_1})a^*(\rmu_{t,x_2})a(\rmu_{t,z_2}) \nonumber\\
		&
		\qquad + a^*(\bar{\rmv}_{t,z_1}) a(\bar{\rmv}_{t,x_1}) a^*(\bar{\rmv}_{t,z_2}) a^*(\rmu_{t,x_2}) - a^*(\rmu_{t,x_1})a(\rmu_{t,z_1})a(\bar{\rmv}_{t,x_2})a(\rmu_{t,z_2}) + a^*(\rmu_{t,x_1}) a(\rmu_{t,z_1})  a^*(\bar{\rmv}_{t,z_2}) a(\bar{\rmv}_{t,x_2}) \nonumber\\
		&
		\qquad + a^*(\rmu_{t,x_1})a^*(\bar{\rmv}_{t,z_1})a(\bar{\rmv}_{t,x_2})a(\rmu_{t,z_2}) + a^*(\bar{\rmv}_{t,z_1}) a^*(\rmu_{t,x_1}) a^*(\bar{\rmv}_{t,z_2})a(\bar{\rmv}_{t,x_2}) + a^*(\rmu_{t,x_1})  a(\rmu_{t,z_1}) a^*(\rmu_{t,x_2}) a(\rmu_{t,z_2})\nonumber\\
		&
		\qquad + a^*(\rmu_{t,x_1}) a(\rmu_{t,z_1}) a^*(\rmu_{t,x_2}) a^*(\bar{\rmv}_{t,z_2})  + a^*(\rmu_{t,x_1})a^*(\bar{\rmv}_{t,z_1})a^*(\rmu_{t,x_2})a^*(\bar{\rmv}_{t,z_2}) +  a^*(\bar{\rmv}_{t,z_1})a(\bar{\rmv}_{t,x_1})a^*(\bar{\rmv}_{t,z_2})a(\bar{\rmv}_{t,x_2})\nonumber\\
		&
		\qquad - a^*(\bar{\rmv}_{t,z_1}) a(\bar{\rmv}_{t,x_1}) a(\rmu_{t,z_2}) a(\bar{\rmv}_{t,x_2})  + a^*(\rmu_{t,x_1})a^*(\bar{\rmv}_{t,z_1}) a^*(\rmu_{t,x_2}) a(\rmu_{t,z_2}) - a(\rmu_{t,z_1}) a(\bar{\rmv}_{t,x_1}) a^*(\bar{\rmv}_{t,z_2}) a(\bar{\rmv}_{t,x_2})\nonumber\\
		&
		\qquad
		- \left<\bar{\rmv}_{t,x_2}, \bar{\rmv}_{t,z_1} \right> a(\bar{\rmv}_{t,x_1}) a(\rmu_{t,z_2})
		+ \left<\bar{\rmv}_{t,x_1}, \bar{\rmv}_{t,z_1} \right> a(\bar{\rmv}_{t,x_2}) a(\rmu_{t,z_2}) \nonumber\\
		&
		\qquad- \left<\rmu_{t,z_1}, \rmu_{t,x_2} \right> a(\bar{\rmv}_{t,x_1}) a(\rmu_{t,z_2})
		+ \left<\rmu_{t,z_1}, \rmu_{t,x_2} \right> a^*(\bar{\rmv}_{t,z_2}) a(\bar{\rmv}_{t,x_1})  -  \left<\rmu_{t,z_1}, \rmu_{t,x_2} \right> \left<\bar{\rmv}_{t,x_1}, \bar{\rmv}_{t,z_2} \right>\nonumber\\
		&
		\qquad
		+ \left<\bar{\rmv}_{t,x_1}, \bar{\rmv}_{t,z_1} \right> a^*(\rmu_{t,x_2})  a(\rmu_{t,z_2})
		+ \left<\bar{\rmv}_{t,x_1}, \bar{\rmv}_{t,z_1} \right> a^*(\rmu_{t,x_2})  a^*(\bar{\rmv}_{t,z_2})  \nonumber\\
		&
		\qquad
		- \left<\bar{\rmv}_{t,x_2}, \bar{\rmv}_{t,z_1} \right>  a^*(\rmu_{t,x_1})a(\rmu_{t,z_2})
		- \left<\bar{\rmv}_{t,x_2}, \bar{\rmv}_{t,z_1} \right> a^*(\rmu_{t,x_1})a^*(\bar{\rmv}_{t,z_2}) \nonumber\\
		&
		\qquad  - \left<\rmu_{t,z_1}, \rmu_{t,x_2} \right> a^*(\rmu_{t,x_1})a(\rmu_{t,z_2})
		- \left<\rmu_{t,z_1},\rmu_{t,x_2} \right> a^*(\rmu_{t,x_1})a^*(\bar{\rmv}_{t,z_2}) \nonumber\\
		&
		\qquad  - \left<\bar{\rmv}_{t,x_2}, \bar{\rmv}_{t,z_2}\right> a^*(\bar{\rmv}_{t,z_1}) a(\bar{\rmv}_{t,x_1}) + \left<\bar{\rmv}_{t,x_2}, \bar{\rmv}_{t,z_1}\right> a^*(\bar{\rmv}_{t,z_2}) a(\bar{\rmv}_{t,x_1}) \nonumber\\
		&
		\qquad +\left< \bar{\rmv}_{t,x_2}, \bar{\rmv}_{t,z_2} \right> a^*(\rmu_{t,x_1}) a(\rmu_{t,z_1}) - \left< \bar{\rmv}_{t,x_2}, \bar{\rmv}_{t,z_2} \right> a^*(\bar{\rmv}_{t,z_1}) a^*(\rmu_{t,x_1})\nonumber\\
		&
		\qquad   - \left<\bar{\rmv}_{t,x_1}, \bar{\rmv}_{t,z_1}\right> a^*(\bar{\rmv}_{t,z_2}) a(\bar{\rmv}_{t,x_2})  -  {\left<\bar{\rmv}_{t,x_2}, \bar{\rmv}_{t,z_1}\right> \left<\bar{\rmv}_{t,x_1}, \bar{\rmv}_{t,z_2}\right>} \bigg] \cU_N(t;0)  \xi_N  \bigg>\nonumber\\
		&
		= : \sum_{i=1}^{16} A_i +  \sum_{j=1}^{16} B_j + {C}. \numberthis \label{many_terms}
	\end{align}
	Using the fact that $\opnorm{\rmu_t},\opnorm{\rmv_t} \leq 1$, $\hsnorm{\rmv_t} \leq \sqrt{N}$, $\trnorm{\omega_{N,t}} = N$ and the assumption $\norm{\xi_N} \leq 1$, we do the following estimates for the first term from $\{A_i\}_{i=1}^{16}$ and $\{B_i\}_{i=1}^{16}$ separately.
	\begin{align*}
		&|A_1|\\
		&
		= \bigg| \int \dx_1 \dx_2  \mathrm{d}z_1 \mathrm{d}z_2\ \big<\xi_N, U^*_N(t;0)    \int  \deta_1 \deta'_1\ a_{\eta_1} a_{\eta'_1} \rmv_t(\eta_1;x_1) O_1(x_1;z_1) {\rmu_t(z_1;\eta'_1)}\\
		&
		\qquad  \int \deta_2 \deta'_2\ a_{\eta_2} a_{\eta'_2} \rmv_t(\eta_2;x_2) O_2(x_2; z_2) {\rmu_t(z_2;\eta'_2)} \cU_N(t;0) \xi_N \big> \bigg|\\
		%		&
		%		= \bigg| \big<\xi_N, U^*_N(t;0)    \int  \deta_1 \deta'_1\ a_{\eta_1} a_{\eta'_1} \big(\rmv_t O_1 \rmu_t\big)(\eta_1;\eta'_1) \int \deta_2 \deta'_2\ a_{\eta_2} a_{\eta'_2} \big(\rmv_t O_2 \rmu_t\big)(\eta_2;\eta'_2) \cU_N(t;0) \xi_N \big> \bigg|\\
		&
		= \bigg|\int \deta_1 \deta'_2 \big< \xi_N,  U^*_N(t;0)  a_{\eta_1} a\big(\overline{\rmv_t O_1 \rmu_t}(\eta_1;\cdot)\big)  a\big(\overline{\rmv_t O_2 \rmu_t}(\cdot;\eta'_2)\big) a_{\eta'_2}  \cU_N(t;0) \xi_N \big> \bigg|\\
		%		&
		%		= \bigg|\int \deta_1 \deta'_2 \big< a^*\big(\overline{\rmv_t O_1 \rmu_t}(\eta_1;\cdot)\big) a^*_{\eta_1} \cU_N(t;0)  \xi_N,   a\big(\overline{\rmv_t O_2 \rmu_t}(\cdot;\eta'_2)\big) a_{\eta'_2}  \cU_N(t;0) \xi_N \big> \bigg|\\
		&
		\leq \int \dd{\eta_1}  \norm{ a^*\big(\overline{\rmv_t O_1 \rmu_t}(\eta_1;\cdot)\big) a^*_{\eta_1} \cU_N(t;0)  \xi_N} \int \dd{\eta'_2} \norm{ a\big(\overline{\rmv_t O_2 \rmu_t}(\cdot;\eta'_2)\big) a_{\eta'_2}  \cU_N(t;0) \xi_N } \\
		&
		\leq \int \dd{\eta_1}  \norm{\rmv_t O_1 \rmu_t(\eta_1;\cdot)}_2 \norm{a^*_{\eta_1} \cU_N(t;0)  \xi_N} \int \dd{\eta'_2} \norm{ \rmv_t O_2 \rmu_t(\cdot;\eta'_2)}_2 \norm{a_{\eta'_2}  \cU_N(t;0) \xi_N }\\
		&
		\leq  \hsnorm{\rmv_t O_1 \rmu_t}  \left(\int \dd{\eta_1} \norm{a^*_{\eta_1} \cU_N(t;0)  \xi_N}^2 \right)^\frac{1}{2} \hsnorm{\rmv_t O_2 \rmu_t} \left(\int \dd{\eta'_2} \norm{a_{\eta'_2}  \cU_N(t;0) \xi_N }^2 \right)^\frac{1}{2}\\
		%		&
		%		{=} \hsnorm{\rmv_t O_1 \rmu_t} \left(\int \dd{\eta_1}  \left<\xi_N, \cU_N^*(t;0)  a_{\eta_1} a^*_{\eta_1} \cU_N(t;0)  \xi_N \right> \right)^\frac{1}{2} \\
		%		&
		%		\qquad \hsnorm{\rmv_t O_2 \rmu_t}  \left(\int \dd{\eta'_2}  \left<\xi_N, \cU_N^*(t;0)  a^*_{\eta'_2} a_{\eta'_2} \cU_N(t;0)  \xi_N \right> \right)^\frac{1}{2}\\
		%		&
		%		\leq \sqrt{N} \hsnorm{\rmv_t O_1} \opnorm{\rmu_t} \hsnorm{\rmv_t O_2} \opnorm{\rmu_t}  \norm{\left(\cN + 1 \right)^\frac{1}{2} \cU_N(t;0) \xi_N}\\
		&
		\leq \sqrt{N} \opnorm{\rmv_t} \hsnorm{O_1} \opnorm{\rmu_t} \hsnorm{\rmv_t} \opnorm{O_2} \opnorm{\rmu_t}  \norm{\left(\cN + 1 \right)^\frac{1}{2} \cU_N(t;0) \xi_N}\\
		&
		\leq N \hsnorm{O_1} \opnorm{O_2} \norm{\left(\cN + 1 \right)^\frac{1}{2} \cU_N(t;0) \xi_N}.
	\end{align*}
	Additionally, we have
	\begin{align*}
		&|B_1|\\
		&
		= \bigg| \int \dx_1 \dx_2 \mathrm{d}z_1 \mathrm{d}z_2\ O_1(x_1;z_1) O_2(x_2;z_2)  \big<  \xi_N, \cU^*_N(t;0) , \left<\bar{\rmv}_{t,x_2}, \bar{\rmv}_{t,z_1} \right> a(\bar{\rmv}_{t,x_1}) a(\rmu_{t,z_2})    \cU_N(t;0)  \xi_N \big> \bigg|\\
		%		&
		%		= \bigg| \int \dx_1 \dx_2 \mathrm{d}z_1 \mathrm{d}z_2\ O_1(x_1;z_1) O_2(x_2;z_2)  \big<  \xi_N, \cU^*_N(t;0) , \overline{\omega_{N,t}(x_2;z_1)} \\
		%		&
		%		\qquad \int  \deta \deta'\ a_{\eta} a_{\eta'} \rmv_t(\eta;x_1) \overline{\rmu_t(\eta'; z_2)}   \cU_N(t;0)  \xi_N \big> \bigg|\\
		&
		= \bigg| \int \dx_1 \dx_2 \mathrm{d}z_1 \mathrm{d}z_2 \big<  \xi_N, \cU^*_N(t;0) , \int  \deta \deta'\ a_{\eta}    a_{\eta'} \\
		&
		\qquad \rmv_t(\eta;x_1) O_1(x_1;z_1) \omega_{N,t}(z_1;x_2)  O_2(x_2;z_2) {\rmu_t(z_2;\eta')}   \cU_N(t;0)  \xi_N \big> \bigg|\\
		&
		= \bigg|  \int  \deta \deta' \big<  \xi_N, \cU^*_N(t;0) , a_{\eta}    a_{\eta'}   \big(\rmv_t O_1 \omega_{N,t} O_2 \rmu_t\big)(\eta;\eta')   \cU_N(t;0)  \xi_N \big> \bigg|\\
		&
		\leq \hsnorm{\rmv_t O_1 \omega_{N,t} O_2 \rmu_t} \norm{\xi_N} \norm{\cN^{1/2} \cU_N(t;0)  \xi_N}\\
		%		&
		%		\leq  \opnorm{\rmv_t} \hsnorm{O_1} \opnorm{\omega_{N,t}} \opnorm{O_2} \opnorm{\rmu_t} \norm{\cN^{1/2} \cU_N(t;0)  \xi_N}\\
		&
		\leq  \hsnorm{O_1} \opnorm{O_2} \norm{\cN^{1/2} \cU_N(t;0)  \xi_N}.
	\end{align*}
	
	The estimates for the rest of the terms can be done with similar steps
	After getting the bound of each terms, we have obtained the following estimates:
	\begin{equation}
		\begin{aligned}
			\left|\sum_{i=1}^{16} A_i \right| &\leq  N \hsnorm{O_1} \opnorm{O_2}  \norm{\cN \cU_N(t;0) \xi_N}, \\
			\left|\sum_{j=1}^{16} B_i \right| &\leq  \hsnorm{O_1} \opnorm{O_2} \bigg(N \norm{ (\cN+1)^{1/2} \cU_N(t;0)  \xi_N} + \sqrt{N}  \norm{ (\cN+1) \cU_N(t;0)  \xi_N}   \bigg)\\
			&
			\leq   N  \hsnorm{O_1} \opnorm{O_2} \norm{ (\cN+1) \cU_N(t;0)  \xi_N}.
		\end{aligned}
	\end{equation}
	Lastly, the final term is estimated as follows:
	\begin{align*}
		&|C|\\
		& = \bigg|\int \dx_1 \dx_2 \mathrm{d}z_1 \mathrm{d}z_2\ O_1(x_1;z_1) O_2(x_2;z_2)  \big<  \xi_N, \cU^*_N(t;0)\left<\bar{\rmv}_{t,x_2}, \bar{\rmv}_{t,z_1}\right> \left<\bar{\rmv}_{t,x_1}, \bar{\rmv}_{t,z_2}\right> \cU_N(t;0)  \xi_N \big> \bigg|\\
		%		&
		%		= \bigg|\int \dx_1 \dx_2 \mathrm{d}z_1 \mathrm{d}z_2\ O_1(x_1;z_1) O_2(x_2;z_2)  \big<  \xi_N, \cU^*_N(t;0) \omega_{N,t}(z_1;x_2) \omega_{N,t}(z_2;x_1) \cU_N(t;0)  \xi_N \big> \bigg|\\
		&
		= \bigg|\int  \dd{x_2}   \big<  \xi_N, \cU^*_N(t;0) \big( O_1 \omega_{N,t} O_2 \omega_{N,t} \big) (x_2;x_2) \cU_N(t;0)  \xi_N \big> \bigg|\\
		%		&
		%		= \bigg| \bigg<  \xi_N, \cU^*_N(t;0) \left( \int  \dd{x_2}  \big( \bar{\omega}_{N,t} \overline{O^*_1} \bar{\omega}_{N,t} \overline{O^*_2} \big) (x_2;x_2)\right) \cU_N(t;0)  \xi_N \bigg> \bigg|\\
		&
		\leq \trnorm{O_1 \omega_{N,t} O_2 \omega_{N,t}} \big|\big< \xi_N, \cU^*_N(t;0) \cU_N(t;0) \xi_N \big> \big|\\
		%		&
		%		\leq \hsnorm{O_1 \omega_{N,t}} \hsnorm{O_2 \omega_{N,t}}\\
		%		&
		%		\leq \hsnorm{O_1} \opnorm{\omega_{N,t}} \opnorm{O_2} \hsnorm{\omega_{N,t}}\\
		&
		\leq \sqrt{N} \hsnorm{O_1} \opnorm{O_2}.
	\end{align*}
	As a summary, we have
	\begin{equation}\label{eq:estimated_mixed}
		\begin{split}
			\left|\Tr O (\gamma_{N,t}^{(2)}- \omega_{N,t} \otimes \omega_{N,t} )\right| & \leq \bigg|\sum_{i=1}^{16} A_i\bigg| + \bigg|\sum_{j=1}^{16} B_j\bigg| + |C|\\
			&
			\leq N \hsnorm{O_1} \opnorm{O_2} \norm{(\cN+1) \cU_N(t;0) \xi_N },
		\end{split}
	\end{equation}
	which implies that, for $O_1$ and $O_2$ being Hilbert-Schmidt and trace class operators, we get
	\begin{equation}
		\begin{aligned}
			&\left(\int \dx_1\dy_1\ \left[ \int \dd{x_2} \left|\gamma_{N,t}^{(2)}(x_1,x_2;y_1,x_2)  - \omega_{N,t}(x_1;y_1)\omega_{N,t}(x_2;x_2) \right| \right]^2  \right)^\frac{1}{2}\\
			&
			\qquad \leq  N  \norm{(\cN+1) \cU_N(t;0) \xi_N } .
		\end{aligned}
	\end{equation}
	Applying Lemma \ref{lem:est_U}, we obtain the inequalities in Proposition \ref{prop:mixed_norm} as desired.
\end{proof}

Finally, we have the following estimate for the mixed-norm. Since it is one of the main contributions of this paper, we write it as a theorem.

\begin{Theorem}\label{thm:mixed_norm_factorized}
	Suppose the assumptions given in Proposition \ref{prop:mixed_norm} hold. Then, we have the following  estimate
	\begin{equation}\label{eq:col_mixed}
		\left(\int \dw_1\du_1\ \left[ \mathds{T}^{(1)} \big|\gamma_{N,t}^{(2)}  - \gamma^{(1)}_{N,t} \otimes \gamma^{(1)}_{N,t}  \big| (u_1;w_1) \right]^2  \right)^\frac{1}{2} \leq  C_t N,
	\end{equation}
	where the constant $C_t$ depends on potential $V$ and time $t$.
\end{Theorem}

\begin{proof}[Proof of Theorem \ref{thm:mixed_norm_factorized}]
	Inserting the intermediate terms
	\begin{equation}\label{eq:insert}
		\begin{aligned}
			&\gamma_{N,t}^{(2)}(u_1,u_2;w_1,w_2)  - \gamma_{N,t}^{(1)}(u_1;w_1)\gamma_{N,t}^{(1)}(u_2;w_2) \\
			&
			=  \gamma_{N,t}^{(2)}(u_1,u_2;w_1,w_2) - \omega_{N,t}(u_1;w_1) \omega_{N,t}(u_2;w_2)\\
			&
			\quad + \big[ \omega_{N,t}(u_1;w_1)  - \gamma^{(1)}_{N,t}(u_1;w_1)\big] \omega_{N,t}(u_2;w_2)\\
			&
			\quad + \gamma^{(1)}_{N,t}(u_1;w_1) \big[\omega_{N,t}(u_2;w_2) -  \gamma_{N,t}^{(1)}(u_2;w_2)\big]\\
			&
			=: T_1 + T_2 +T_3,
		\end{aligned}
	\end{equation}
	the estimate in \eqref{eq:col_mixed} is then reduced into the estimates of the following terms.
	\begin{itemize}
		\item $\left(\displaystyle\int \dw_1\du_1\ \left[ \mathds{T}^{(1)} \big|\gamma_{N,t}^{(2)}  - \omega^{(1)}_{N,t} \otimes \omega^{(1)}_{N,t}  \big| (u_1;w_1) \right]^2   \right)^\frac{1}{2}$,
		
		\item $ \hsnorm{\omega_{N,t} -  \gamma_{N,t}^{(1)}} \trnorm{\omega_{N,t}}$,
		
		\item $\hsnorm{\gamma_{N,t}^{(1)}} \trnorm{\omega_{N,t} -  \gamma_{N,t}^{(1)}}$.
	\end{itemize}
	Proposition \ref{prop:mixed_norm} implies immediately \eqref{eq:col_mixed} considering the fact that $\trnorm{\omega_{N,t}} \leq N$.
\end{proof}

\begin{proof}[Proof of Proposition \ref{prop:est_mean-field_osc}]
	Direct corollary from Theorem \ref{thm:mixed_norm_factorized}.
\end{proof}

\section{Proof of Theorem \ref{thm:main_ben}}\label{sec:proof_mainTheorem}
{{
		In this section, we prove the main theorem. As have been mentioned in the introduction.	
		We will show that, for any $T>0$, the sequence $m_{N,t}$ is weakly compact 
		and that any accumulation point $m_t$ is exactly the solution of the Vlasov equation.
		For this purpose, we need the following lemma.
		\begin{Lemma}
			Let $m_{N,t}$ be weak solution of the reformulated Schr\"odinger equation \eqref{eq:BBGKY_k1_ben} and $\rho_{N,t}(q):=\int dp\, m_{N,t}(p,q)$.
			Then there exists a subsequence of $m_{N,t}$ (without relabeling for convenience) and function $m_t$ such that as $N\rightarrow \infty$
			\begin{align}
				&m_{N,t}\stackrel{*}{\rightharpoonup} m_t\quad {\rm{in}}\quad L^\infty(0,T;L^s(\bR^3\times\bR^3)),\quad s\in [1,\infty],\label{lN1}\\
				&\nabla V\ast \rho_{N,t}\rightarrow \nabla V\ast\rho_t \quad {\rm{in}}\quad L^r(0,T;L^r(\bR^3)),\quad r\in(1,\infty),\label{lN2}
			\end{align}
			where $\rho_t=\int \mathrm{d}p\, m_t$.
		\end{Lemma}
		\begin{proof}
	The estimates in Appendix \ref{sec:apriori} imply
	\begin{align}\label{lN3}
		\|m_{N,t}\|_{L^\infty(0,T;L^1(\bR^3\times\bR^3))}
		+\|m_{N,t}\|_{L^\infty(0,T;L^\infty(\bR^3\times\bR^3))}\leq C,
	\end{align}
	where $C$ appeared in this section denotes a positive constant independent of $N$.
	And combining interpolation inequality, we have
	\begin{align*}
		\|m_{N,t}\|_{L^\infty(0,T;L^s(\bR^3\times\bR^3))}\leq C,\quad s\in[1,\infty].
	\end{align*}
	Therefore, \eqref{lN1} is a direct consequence of the above inequality and the moment estimates in Proposition \ref{prop:2nd_moment_finite}.
	
	Furthermore, Proposition \ref{prop:2nd_moment_finite} implies
	$
	\||p|^2 m_{N,t}\|_{L^\infty(0,T;L^1(\bR^3\times\bR^3))}\leq C,
	$
	together with \eqref{lN3}, we arrive at
	$
	\|\rho_{N,t}\|_{L^\infty(0,T;L^s(\bR^3))}\leq C
	$
	for $s\in [1,\frac{5}{3} ]$.
	Hence there exists a subsequence of $rho_{N,t}$ (which is not re-labeled for convenience) such that
	\begin{align*}
		\rho_{N,t}\stackrel{*}{\rightharpoonup}\rho_t \quad {\rm{in}}\quad L^\infty(0,T;L^s(\bR^3)), \quad s\in \big(1,\frac{5}{3}\big].
	\end{align*}
	Owing to $V\in W^{2,\infty}(\bR^3)$ and Young's convolution inequality, we have for a.e. $t\in(0,T)$,
	\begin{align*}
		\|\nabla^2V\ast \rho_{N,t}\|_{L^\infty(\bR^3)}
		\leq \|\nabla^2 V\|_{L^\infty(\bR^3)}\|\rho_{N,t}\|_{L^1(\bR^3)}\leq C.
	\end{align*}
	Similarly, we obtain that
	\begin{align}\label{lN4}
		\|\nabla V\ast \rho_{N,t}\|_{L^\infty(0,T;W^{1,\infty}(\bR^3))}\leq C.
	\end{align}
	By means of \eqref{eq:BBGKY_k1_ben}, we get
	\begin{align*}
		\partial_t\rho_{N,t}&=\partial_t\int \mathrm{d}p\, m_{N,t}(p,q)
		=-\nabla_q\cdot \int \mathrm{d}p \,p\,m_{N,t}(p,q)+\nabla_q\cdot \int \mathrm{d}p\, \widetilde{\cR}.
	\end{align*}
	It is easy to see that $\partial_t(\nabla V\ast\rho_{N,t}$ satisfies the following equation:
	\begin{align*}
		\partial_t(\nabla V\ast\rho_{N,t})=-\nabla_q\cdot \Big(\nabla V\otimes_\ast\int \mathrm{d}p \,p\,m_{N,t}(p,q) \Big)
		+\nabla_q\cdot \Big(\nabla V\otimes_\ast \int \mathrm{d}p\, \widetilde{\cR} \Big),
	\end{align*}
	where $(u\otimes_\ast v)_{ij}=u_i\ast v_j$ for $(u,v)\in \bR^3\times\bR^3$.
	Noticing that 
	\begin{align*}
		\Big\|\int \mathrm{d}p \,p\,m_{N,t}(p,q) \Big\|_{L^\infty(0,T;L^s(\bR^3))}\leq C,\quad s\in \Big[1,\frac{5}{4} \Big],
	\end{align*}
	we derive for any test function $\tilde \varphi(q)\in W^{1,3}(\bR^3)$ and a.e. $t\in(0,T)$,
	\begin{align*}
		&\Big|\int_{\bR^3}\mathrm{d}q\,\nabla_q\cdot \Big(\nabla V\otimes_\ast\int \mathrm{d}p \,p\,m_{N,t}(p,q) \Big)\tilde \varphi(q)\Big|
		\leq &C(V)
		\Big\|\int \mathrm{d}p\,p\,m_{N,t}(p,q) \Big\|_{L^1(\bR^3)}\|\nabla\tilde \varphi(q)\|_{L^{3}(\bR^3)}.
	\end{align*}
	For the second term, applying Proposition \ref{prop:tildeR}, we have
	\begin{align*}
    	\Big|\int \mathrm{d}q\,\nabla_q\cdot \Big(\nabla V\otimes_\ast \int \mathrm{d}p\, \widetilde{\cR} \Big)\tilde\varphi(q) \Big|\leq C(V)\norm{\int \mathrm{d}p\widetilde{\cR}}_{L^{\frac{5}{4}}(\R^3)}\norm{\nabla\varphi(q)}_{L^{3}(\R^3)}\leq C\hbar\norm{\nabla\varphi(q)}_{L^{3}(\R^3)} .
	\end{align*}
	The estimates above show that 
	\begin{align}\label{lN5}
		\|\partial_t(\nabla V\ast \rho_{N,t})\|_{L^\infty(0,T;W^{-1,\frac{3}{2}}(\bR^3))}\leq C.
	\end{align}
	The inequalities \eqref{lN4} and \eqref{lN5} allow us to apply Aubin-Lions lemma (e.g. in \cite{aubin1963analyse,lions1969quelques}) to infer that \eqref{lN2}. We mention here that the application of Aubin-Lions lemma is proceeded in a sequence of growing balls, and the convergent subsequence is obtained through diagonal rule.
\end{proof}

		\begin{proof}[Proof of Theorem \ref{thm:main_ben}]
			With the help of \eqref{lN1}, \eqref{lN2},  Proposition \ref{prop:tildeR}, Proposition \ref{prop:est_R_1} and Proposition \ref{prop:est_mean-field_osc},
			we can take limit $N\rightarrow \infty$ in the weak formulation of the reformulated Schr\"odinger equation \eqref{eq:BBGKY_k1_ben}.
			More precisely, for any $\phi,\varphi\in C_0^\infty(\bR^3)$ and $\eta\in C_0^\infty[0,T)$, $m_{N,t}$ satisfies the following equation
			\begin{align*}
				\int_0^T&\dt\int \dq\ddp\, m_{N,t}(p,q)\big[\partial_t\eta\varphi(q)\phi(p)+p\cdot\eta(t)\nabla_q\varphi(q)\phi(p)
				-{\frac{1}{(2\pi)^3}}\nabla V\ast\rho_{N,t}\eta(t)\varphi(q)\cdot\nabla_p\phi(p) \big]\\
				&=\int_0^T \dt\,\eta(t)\int \dq\ddp\,\big(\nabla \varphi(q)\,\phi(p)\,\widetilde{\cR}
				+ \varphi(q)\,\nabla\phi(p)(\cR_\rms + \cR_\rmm)\big) - \eta(0)\int \dq\ddp\,\varphi(q)\phi(p)m_{{N},0}.
			\end{align*}
			Since the sums and products of functions of the
			form $\eta(t)\varphi(q)\phi(p)$ are dense in $C_0^\infty([0,T)\times\bR^3\times\bR^3)$, we have showed that the limit of the subsequence is a weak solution of the Vlasov equation. On the other hand, the assumption 1 in \ref{assumptions} implies that $V\in W^{2,\infty}$, from which one obtains that the Vlasov equation has a unique weak solution {as in \cite[Theorem 1.1]{Golse2022}.} Therefore, the whole sequence $m_{N,t}$ converges weakly. Hence the proof of Theorem \ref{thm:main_ben} is completed.
		\end{proof}
	}
}
%\section*{Data Availability}
%Data sharing not applicable to this article as no datasets were generated or analyzed during the current study.
%
%\section*{Conflicts of interest}
%The authors have no competing interests to declare that are relevant to the content of this article.

\bibliographystyle{abbrv}
\bibliography{husimi_vlasov}

\clearpage
\begin{appendices}
	\section{Second Quantization}\label{sec:FockSpace}
	The Fock space formalism and some results from Bogoliubov {{theory}} for the proof of this paper are listed in the following.\footnote{See \cite{nam_qm2} for more pedagogic treatment on the topics.} In particular, as in \cite{Benedikter2016book}, we will introduce the fermionic Fock space over Hilbert spaces as the following direct sum:
	\[
	\mathcal{F}_a := \bC {\oplus} \bigoplus_{n \geq 1} L^2_a (\R^{3n}).
	\]
	{By convention, we say that the vacuum state, denoted as $\Omega = \{1,0,0,\dots \}$, belongs to $\bC$.} For all $\Psi = \{\psi^{(n)}\}_{n\in \bN} \in \mathcal{F}_a$ and $\psi^{(n)} \in L^2_a (\R^{3n})$, we define the number of particle operator on the $n$-th sector by $\big(\cN \Psi \big)^{(n)} = n \psi^{(n)}$.
	
	As in \cite{Benedikter2016book}, the creation and annihilation operators acting on $\Psi \in \mathcal{F}_a$ is defined as follows: for any $f \in L^2(\R^3)$
	\begin{align*}
		\big(a^*(f)\Psi \big)^{(n)} (x_1, \dots, x_n) &:= \sum_{j=1}^n \frac{(-1)^j}{\sqrt{n}} f(x_j) \psi^{(n-1)} (x_1,\dots,x_{j-1}, x_{j+1},\dots,x_n ),\\
		\big(a(f)\Psi \big)^{(n)} (x_1, \dots, x_n) &:= \sqrt{n+1} \int \dd{x} \overline{f(x)} \psi^{(n+1)}(x,x_1,\dots,x_n),
	\end{align*}
	where $\psi^{(n)} \in L^2(\R^{3n})$ for any $n\in \bN$. Additionally, for convenient purposes, the creation and annihilation operators will be represented by its operator-value distribution, $a^*_x$ and $a_x$ respectively, so that
	\[
	a^*(f) = \int \dd{x} f(x) a_x, \quad a(f) = \int \dd{x} \overline{f(x)} a_x.
	\]
	Therefore, the canonical anticommutator relation(CAR) is written as
	\[
	\{a^*_x, a_y\} = \delta_{x=y}, \quad \{a^*_x, a^*_y\} = \{a_x, a_y\} = 0,
	\]
	for any $x,y \in \R^3$.
	
	Observe that for given any $\Psi, \Phi \in \mathcal{F}_a$, it holds that
	\[
	\left< \Psi, \cN \Phi \right> = \int \mathrm{d}{x} \left< a_x \Psi, a_x \Phi \right>.
	\]
	Therefore, we write the number of particles operator as $\cN = \int \dd{x} a^*_x a_x$. Similarly, the integral kernel of the $k$-particle reduced density matrix is written as follows:
	\begin{equation}\label{eq:k_reduced}
		\gamma^{(k)} (x_1, \dots, x_k; y_1, \dots, y_k) = \left< \Psi, a^*_{y_1} \dots a^*_{y_k} a_{x_k} \cdots a_{x_1}  \Psi \right>.
	\end{equation}
	Moreover, the Hamiltonian acting on $\Psi \in \mathcal{F}_a$ can be written as
	\begin{equation}\label{eq:fock_Hamil}
		\mathcal{H}_N := \frac{\hbar^2}{2}\int \mathrm{d}{x} \nabla_x a^*_x \nabla_x a_x + \frac{1}{2N} \int \mathrm{d}{x}\mathrm{d}{y}\, V(x-y) a^*_x a^*_y a_y a_x,
	\end{equation}
	where $V$ is the interaction potential. We will denote the operator of the kinetic term as
	\begin{equation}\label{eq:fock_kinetic}
		\mathcal{K} = \hbar^2 \int \dd{x} \nabla_x a^*_x \nabla_x a_x.
	\end{equation}
	
	As presented in \cite{benedikter2014mean, solo_lect}, for any $t \geq 0$, there exists a unitary transformation $\cR_{\mathcal{V}_{N,t}}: \mathcal{F}_a \to \mathcal{F}_a$ such that
	\begin{equation}\label{eq:exciting_anil_creation}
		\begin{split}
			\cR^*_{\mathcal{V}_{N,t}} a_x \cR_{\mathcal{V}_{N,t}} &= a(\rmu_{t,x}) + a^*(\bar{\rmv}_{t,x}),\\
			\cR^*_{\mathcal{V}_{N,t}} a_x^* \cR_{\mathcal{V}_{N,t}} &= a^*(\rmu_{t,x}) + a(\bar{\rmv}_{t,x}),
		\end{split}
	\end{equation}
	where  $\rmv_{t,x} := \sum_{j=1}^N \dyad{\overline{\rme}_{j,t} }{\rme_{j,t}}$ and $\rmu_{t,x} := \mathds{1} - \sum_{j=1}\dyad{{\rme}_{j,t} }{\rme_{j,t}}$, for any orthonormal basis $\{\rme_{j,t}\}_{j=1}^N \subset L^2(\R^3)$.
	
	Then, for $t \geq 0$, the solution of the Schr\"odinger equation is given as
	\begin{equation}\label{eq:psi_bogo}
		\Psi_{N,t} = e^{-\frac{\ii}{\hbar}\mathcal{H}_Nt} \cR_{\mathcal{V}_{N,0}}\Omega = \cR_{\mathcal{V}_{N,t}}\cU_{N}(t;0)  \Omega ,
	\end{equation}
	where $\cR_{\mathcal{V}_{N,t}}$ is a unitary Bogoliubov mapping and $\cU_{N}$ is the quantum fluctuation dynamics defined as follows,
	\begin{equation}\label{eq:flunctuation}
		\cU_{N} (t;s) := R_{\mathcal{V}_{N,t}}^*  e^{-\frac{\ii}{\hbar}\mathcal{H}_N(t-s)} R_{\mathcal{V}_{N,s}}.
	\end{equation}
	
	\section{\textit{A priori} estimates}\label{sec:apriori}
	
	In this appendix, we present in this section a sequence of estimates from \cite{Chen2021JSP} that will prove useful to our calculation. First, we have the following properties of $k$-particle Husimi measures from \cite[Lemma 2.2]{Chen2021JSP}
	
	\begin{Lemma} \label{prop:kHusimi}
		Let $m^{(k)}_{N,t}$ be the $k$-particle Husimi measure as defined in \eqref{eq:husimi_def_1}. Then, the following properties hold:
		\begin{enumerate}
			\item  $m^{(k)}_{N,t}(q,p,\dots,q_k,p_k)$ is symmetric,
			\item  $\frac{1}{(2\pi)^{3k}} \int (\dq\ddp)^{\otimes k} m^{(k)}_{N,t}(q,p,\dots,q_k,p_k) = \frac{N(N-1)\cdots (N-k+1)}{N^k}$,
			\item $\frac{1}{(2\pi \hbar)^{3}} \int \dq_k  \ddp_k\ m^{(k)}_{N,t}(q,p,\dots,q_k,p_k) = (N-k+1) m^{(k-1)}_{N,t}(q,p,\dots,q_{k-1},p_{k-1}) $,
			\item $ 0 \leq  m^{(k)}_{N,t}(q,p,\dots,q_k,p_k) \leq 1$ a.e.,
		\end{enumerate}
		where $1 \leq k \leq N$.
	\end{Lemma}
	
	From \cite[Lemma 2.6]{Chen2021JSP} and \cite[Proposition 2.3]{Chen2021JSP}, we {have} the following estimate for the kinetic energy as well as the moment estimate of the $1$-particle Husimi measure respectively:
	
	\begin{Lemma}\label{lem:kinetic_finite}
		Assume $V \in W^{1,\infty} (\R^3)$, then the kinetic energy is bounded as follows:
		\begin{equation}\label{eq:k_kinetic_bounded}
			\left<\Psi_{N,t}, \frac{\mathcal{K}}{N} \Psi_{N,t}\right> \leq \left< \Psi_{N}, \frac{\mathcal{K}}{ N} \Psi_{N} \right> + Ct^2,
		\end{equation}
		where $\mathcal{K}$ is defined in \eqref{eq:fock_kinetic} and the constant $C$ depends on $\norm{\nabla V}_\infty$.
	\end{Lemma}
	
	\begin{Proposition} \label{prop:2nd_moment_finite} For $t \geq 0$, we have the following finite moments:
		\begin{equation}\label{eq:2nd_moment_finite_0}
			\int \dq \ddp\, (|{q}| + |{p}|^2) m_{N,t}(q,p) \leq C(1+t^3),
		\end{equation}
		where $C>0$ is a constant that depends on initial data $\int \dq \ddp\, (|q| + |p|^2) m_{N}(q, p) $.
	\end{Proposition}
	
	Next, we will present the oscillation estimate from \cite[Lemma 2.5]{Chen2021JSP} which will be used frequently in our proof:
	
	\begin{Lemma}[Bound on localized number operator]\label{lem:localized} Let $\psi_N \in \mathcal{F}^{(N)}_a$ such that $\norm{\psi_N} = 1$, and $R$ be the radius of a ball such that the volume is $1$. Then we have
		\begin{align*}
			&\int \dq\dx \left< \psi_N ,
			\rchi_{|x-q|\leq \sqrt{\hbar}R}a^*_{x}a_{x} \psi_N  \right>
			\leq C(R)\hbar^{-\frac{3}{2}},
		\end{align*}
		where $\rchi$ is a characteristic function.
	\end{Lemma}
	
	\begin{Lemma}[Estimate of oscillation]\label{lem:estimate_oscillation} For $\varphi \in C^\infty_0 (\R^3)$ and
		\begin{equation}\label{eq:estimate_oscillation_omega}
			\Omega_\hbar^\alpha := \{x \in \R^3;\ \max_{1\leq j \leq 3} |x_j|\leq \hbar^\alpha \},
		\end{equation}
		it holds for every $\alpha \in (0,1)$, $s \in \N$, and $x \in \R^3\backslash \Omega_\hbar^\alpha$,
		\begin{equation}\label{eq:estimate_oscillation_0}
			\left|\int_{\R^3} \ddp\,  e^{\frac{\rm i}{\hbar}p\cdot x} \varphi(p)\right| \leq C \hbar^{(1-\alpha)s},
		\end{equation}
		where $C$ depends on the compact support and the $C^s$-norm of $\varphi$.
	\end{Lemma}

	\clearpage
	\section{Rest of Proof of Proposition \ref{prop:mixed_norm}}\label{sec:mixednorm-rest}
	
	In order to cater to readers who require more detailed information, we will present the remaining estimations for each term in \eqref{many_terms} in this appendix:
	\begin{align*}
		&|A_2|\\
		& = \bigg| \int \dx_1 \dx_2 \mathrm{d}z_1 \mathrm{d}z_2\ O_1(x_1;z_1) O_2(x_2;z_2)  \big<  \xi_N, \cU^*_N(t;0)  a(\bar{\rmv}_{t,x_1}) a(\rmu_{t,z_1}) a^*(\rmu_{t,x_2}) a(\rmu_{t,z_2})   \cU_N(t;0)  \xi_N \big> \bigg|\\
%	   &
%		= \bigg| \int \dx_1 \dx_2  \mathrm{d}z_1 \mathrm{d}z_2\ \big<\xi_N, U^*_N(t;0) O_1(x_1;z_1) O_2(x_2; z_2)  \int  \deta_1 \deta'_1\ a_{\eta_1} a_{\eta'_1} \rmv_t(\eta_1;x_1) \overline{\rmu_t(\eta'_1; z_1)}\\
%		&
%		\qquad  \int \deta_2 \deta'_2\ a^*_{\eta_2} a_{\eta'_2} \rmu_t(\eta_2;x_2) \overline{\rmu_t(\eta'_2;z_2)} \cU_N(t;0) \xi_N \big> \bigg|\\
		&
		= \bigg| \int \dx_1 \dx_2  \mathrm{d}z_1 \mathrm{d}z_2\ \big<\xi_N, U^*_N(t;0)   \int  \deta_1 \deta'_1\ a_{\eta_1} a_{\eta'_1} \rmv_t(\eta_1;x_1) O_1(x_1;z_1) {\rmu_t(z_1; \eta'_1)}\\
		&
		\qquad  \int \deta_2 \deta'_2\ a^*_{\eta_2} a_{\eta'_2} \rmu_t(\eta_2;x_2)  O_2(x_2; z_2) {\rmu_t(z_2;\eta'_2)} \cU_N(t;0) \xi_N \big> \bigg|\\
		&
		= \bigg| \big<\xi_N, U^*_N(t;0)   \int  \deta_1 \deta'_1\ a_{\eta_1} a_{\eta'_1} \big(\rmv_t O_1 \rmu_t\big)(\eta_1; \eta'_1) \int \deta_2 \deta'_2\ a^*_{\eta_2} a_{\eta'_2}  \big(\rmu_t O_2 \rmu_t\big)(\eta_2; \eta'_2) \cU_N(t;0) \xi_N \big> \bigg|\\
		&
		= \bigg|\int \deta_1 \big< a^*\big(\overline{\rmv_t O_1 \rmu_t}(\eta_1; \eta'_1)\big) a^*_{\eta_1} \cU_N(t;0)  \xi_N,    \dd{\Gamma}(\rmu_t O_2 \rmu_t) \cU_N(t;0) \xi_N \big> \bigg|\\
		&
		\leq \int \deta_1 \norm{\rmv_t O_1 \rmu_t}_2 \norm{a^*_{\eta_1} \cU_N(t;0)  \xi_N}  \norm{ \dd{\Gamma}(\rmu_t O_2 \rmu_t) \cU_N(t;0)  \xi_N }\\
		&
		\leq  \hsnorm{\rmv_t O_1 \rmu_t} \left(\int \dd{\eta_1}  \left<\xi_N, \cU_N^*(t;0)  a_{\eta_1} a^*_{\eta_1} \cU_N(t;0)  \xi_N \right> \right)^\frac{1}{2} \opnorm{\rmu_t O_2 \rmu_t} \norm{\cN \cU_N(t;0)  \xi_N}\\
		&
		\leq \opnorm{\rmv_t} \hsnorm{O_1} \opnorm{\rmu_t} \opnorm{\rmu_t} \opnorm{O_2} \opnorm{\rmu_t} \norm{\cN^\frac{1}{2} \cU_N(t;0)  \xi_N}  \norm{\cN \cU_N(t;0) \xi_N}\\
		&
		\leq  \hsnorm{O_1} \opnorm{O_2}  \norm{\cN \cU_N(t;0) \xi_N}\\
%		&
%		\leq \sqrt{N} \hsnorm{O_1} \opnorm{O_2}  \norm{\cN \cU_N(t;0) \xi_N}.
	\end{align*}
	\begin{align*}
		&|A_3|\\
		& = \bigg| \int \dx_1 \dx_2 \mathrm{d}z_1 \mathrm{d}z_2\ O_1(x_1;z_1) O_2(x_2;z_2)  \big<  \xi_N, \cU^*_N(t;0)  a(\bar{\rmv}_{t,x_1}) a(\rmu_{t,z_1}) a^*(\rmu_{t,x_2})a^*(\bar{\rmv}_{t,z_2})   \cU_N(t;0)  \xi_N \big> \bigg|\\
		&
		\leq \hsnorm{\rmv_t O_1 \rmu_t} \hsnorm{\rmu_t O_2 \bar{\rmv}_t}  \left<\xi_N, \cU_N^*(t;0) \left(\cN + 1 \right)  \cU_N(t;0) \xi_N\right>\\
		&
		\leq N \hsnorm{O_1} \opnorm{O_2}    \norm{\left(\cN + 1 \right) \cU_N(t;0) \xi_N}.
	\end{align*}
	\begin{align*}
		&|A_4|\\
		& = \bigg| \int \dx_1 \dx_2 \mathrm{d}z_1 \mathrm{d}z_2\ O_1(x_1;z_1) O_2(x_2;z_2)  \big<  \xi_N, \cU^*_N(t;0) a^*(\bar{\rmv}_{t,z_1}) a(\bar{\rmv}_{t,x_1})a^*(\rmu_{t,x_2})a(\rmu_{t,z_2})   \cU_N(t;0)  \xi_N \big> \bigg|\\
		&
		\leq \opnorm{\bar{\rmv}_t \overline{O^*_1} \rmv_t} \opnorm{\rmu_t O_2 \rmu_t}  \left<\xi_N, \cU_N^*(t;0) \cN^2   \cU_N(t;0) \xi_N\right>\\
		&
		\leq \hsnorm{O_1} \opnorm{O_2}  \left<\xi_N, \cU_N^*(t;0) \cN^2  \cU_N(t;0) \xi_N\right>\\
		&
		\leq N \hsnorm{O_1} \opnorm{O_2}  \norm{\left(\cN + 1 \right) \cU_N(t;0) \xi_N}.
	\end{align*}
	\begin{align*}
		&|A_5|\\
		& = \bigg| \int \dx_1 \dx_2 \mathrm{d}z_1 \mathrm{d}z_2\ O_1(x_1;z_1) O_2(x_2;z_2)  \big<  \xi_N, \cU^*_N(t;0) a^*(\bar{\rmv}_{t,z_1}) a(\bar{\rmv}_{t,x_1}) a^*(\bar{\rmv}_{t,z_2}) a^*(\rmu_{t,x_2})   \cU_N(t;0)  \xi_N \big> \bigg|\\
		&
		\leq \hsnorm{\bar{\rmv}_t \overline{O^*_1} {\rmv}_t} \hsnorm{\bar{\rmv}_t \overline{O^*_2} \bar{\rmu}_t}  \left<\xi_N, \cU_N^*(t;0) \left(\cN + 1 \right) \cU_N(t;0) \xi_N\right>\\
		&
		\leq N \hsnorm{O_1} \opnorm{O_2}\left<\xi_N, \cU_N^*(t;0) \left(\cN + 1 \right)   \cU_N(t;0) \xi_N\right>\\
		&
		\leq N \hsnorm{O_1} \opnorm{O_2}  \norm{\left(\cN + 1 \right) \cU_N(t;0) \xi_N}.
	\end{align*}
	\begin{align*}
		&|A_6|\\
		& = \bigg| \int \dx_1 \dx_2 \mathrm{d}z_1 \mathrm{d}z_2\ O_1(x_1;z_1) O_2(x_2;z_2)  \big<  \xi_N, \cU^*_N(t;0) a^*(\rmu_{t,x_1})a(\rmu_{t,z_1})a(\bar{\rmv}_{t,x_2})a(\rmu_{t,z_2})    \cU_N(t;0)  \xi_N \big> \bigg|\\
		&
		\leq \opnorm{\rmu_t O_1 \rmu_t} \hsnorm{{\rmv}_t O_2 {\rmu}_t}  \left<\xi_N, \cU_N^*(t;0) \cN^2 \cU_N(t;0) \xi_N\right>^\frac{1}{2} \left<\xi_N, \cU_N^*(t;0) \cN \cU_N(t;0) \xi_N\right>^\frac{1}{2}\\
		&
		\leq  N \hsnorm{O_1} \opnorm{O_2}  \norm{\cN \cU_N(t;0) \xi_N}.
	\end{align*}
	\begin{align*}
		&|A_7|\\
		& = \bigg| \int \dx_1 \dx_2 \mathrm{d}z_1 \mathrm{d}z_2\ O_1(x_1;z_1) O_2(x_2;z_2)  \big<  \xi_N, \cU^*_N(t;0) a^*(\rmu_{t,x_1}) a(\rmu_{t,z_1})  a^*(\bar{\rmv}_{t,z_2}) a(\bar{\rmv}_{t,x_2}) \cU_N(t;0)  \xi_N \big> \bigg|\\
		&
		\leq \opnorm{\rmu_t O_1 \rmu_t} \opnorm{{\rmv}_t \overline{O^*_2} {\rmv}_t}  \left<\xi_N, \cU_N^*(t;0) \cN^2 \cU_N(t;0) \xi_N\right>\\
		&
		\leq \hsnorm{O_1} \opnorm{O_2} \left<\xi_N, \cU_N^*(t;0) \cN^2   \cU_N(t;0) \xi_N\right>\\
		&
		\leq N \hsnorm{O_1} \opnorm{O_2}  \norm{\cN \cU_N(t;0) \xi_N}.
	\end{align*}
	\begin{align*}
		&|A_8|\\
		& = \bigg| \int \dx_1 \dx_2 \mathrm{d}z_1 \mathrm{d}z_2\ O_1(x_1;z_1) O_2(x_2;z_2)  \big<  \xi_N, \cU^*_N(t;0) a^*(\rmu_{t,x_1})a^*(\bar{\rmv}_{t,z_1})a(\bar{\rmv}_{t,x_2})a(\rmu_{t,z_2}) \cU_N(t;0)  \xi_N \big> \bigg|\\
		&
		\leq \hsnorm{\rmu_t O_1 \bar{\rmv}_t} \hsnorm{{\rmv}_t O_2 {\rmu}_t}  \left<\xi_N, \cU_N^*(t;0) \cN \cU_N(t;0) \xi_N\right>\\
		&
		\leq \sqrt{N} \hsnorm{O_1} \opnorm{O_2} \left<\xi_N, \cU_N^*(t;0) \cN   \cU_N(t;0) \xi_N\right>\\
		&
		\leq N \hsnorm{O_1} \opnorm{O_2}  \norm{\cN^\frac{1}{2}\cU_N(t;0) \xi_N}.
	\end{align*}
	\begin{align*}
		&|A_9|\\
		& = \bigg| \int \dx_1 \dx_2 \mathrm{d}z_1 \mathrm{d}z_2\ O_1(x_1;z_1) O_2(x_2;z_2)  \big<  \xi_N, \cU^*_N(t;0) a^*(\bar{\rmv}_{t,z_1}) a^*(\rmu_{t,x_1}) a^*(\bar{\rmv}_{t,z_2})a(\bar{\rmv}_{t,x_2}) \cU_N(t;0)  \xi_N \big> \bigg|\\
		&
		\leq \hsnorm{\bar{\rmv}_t \overline{O^*_1} \bar{\rmu}_t} \opnorm{\bar{\rmv}_t \overline{O^*_2} {\rmv}_t}  \left<\xi_N, \cU_N^*(t;0) \cN \cU_N(t;0) \xi_N\right>^\frac{1}{2} \left<\xi_N, \cU_N^*(t;0) \cN^2 \cU_N(t;0) \xi_N\right>^\frac{1}{2}\\
		&
		\leq \hsnorm{O_1} \opnorm{O_2} \left<\xi_N, \cU_N^*(t;0) \cN^2   \cU_N(t;0) \xi_N\right>\\
		&
		\leq N \hsnorm{O_1} \opnorm{O_2}  \norm{\cN \cU_N(t;0) \xi_N}.
	\end{align*}
	\begin{align*}
		&|A_{10}|\\
		& = \bigg| \int \dx_1 \dx_2 \mathrm{d}z_1 \mathrm{d}z_2\ O_1(x_1;z_1) O_2(x_2;z_2)  \big<  \xi_N, \cU^*_N(t;0) a^*(\rmu_{t,x_1})  a(\rmu_{t,z_1}) a^*(\rmu_{t,x_2}) a(\rmu_{t,z_2}) \cU_N(t;0)  \xi_N \big> \bigg|\\
		&
		\leq \opnorm{\rmu_t O_1 {\rmu}_t} \opnorm{\rmu_t O_2 \rmu_t}  \left<\xi_N, \cU_N^*(t;0) \cN^2 \cU_N(t;0) \xi_N\right>\\
		&
		\leq \hsnorm{O_1} \opnorm{O_2} \left<\xi_N, \cU_N^*(t;0) \cN^2   \cU_N(t;0) \xi_N\right>\\
		&
		\leq N \hsnorm{O_1} \opnorm{O_2}  \norm{\cN \cU_N(t;0) \xi_N}.
	\end{align*}
	\begin{align*}
		&|A_{11}|\\
		& = \bigg| \int \dx_1 \dx_2 \mathrm{d}z_1 \mathrm{d}z_2\ O_1(x_1;z_1) O_2(x_2;z_2)  \big<  \xi_N, \cU^*_N(t;0) a^*(\rmu_{t,x_1}) a(\bar{\rmv}_{t,z_1})a^*(\rmu_{t,x_2})a^*(\bar{\rmv}_{t,z_2})  \cU_N(t;0)  \xi_N \big> \bigg|\\
		&
		\leq \hsnorm{\rmu_t O_1 \rmu_t} \hsnorm{\rmu_t O_2 \bar{\rmv}_t}  \left<\xi_N, \cU_N^*(t;0) \cN \cU_N(t;0) \xi_N\right>^\frac{1}{2} \left<\xi_N, \cU_N^*(t;0) (\cN + 1 ) \cU_N(t;0) \xi_N\right>^\frac{1}{2}\\
		&
		\leq \sqrt{N} \hsnorm{O_1} \opnorm{O_2} \left<\xi_N, \cU_N^*(t;0)  \left(\cN + 1 \right)  \cU_N(t;0) \xi_N\right>\\
		&
		\leq N \hsnorm{O_1} \opnorm{O_2}  \norm{(\cN + 1 )^\frac{1}{2}\cU_N(t;0) \xi_N}.
	\end{align*}
	\begin{align*}
		&|A_{12}|\\
		& = \bigg| \int \dx_1 \dx_2 \mathrm{d}z_1 \mathrm{d}z_2\ O_1(x_1;z_1) O_2(x_2;z_2)  \big<  \xi_N, \cU^*_N(t;0) a^*(\rmu_{t,x_1})a^*(\bar{\rmv}_{t,z_1})a^*(\rmu_{t,x_2})a^*(\bar{\rmv}_{t,z_2}) \cU_N(t;0)  \xi_N \big> \bigg|\\
		&
		\leq \hsnorm{\rmu_t O_1 {\rmv}_t} \hsnorm{\rmu_t O_2 \bar{\rmv}_t}  \left<\xi_N, \cU_N^*(t;0)  \left(\cN + 1 \right) \cU_N(t;0) \xi_N\right>\\
		&
		\leq \sqrt{N} \hsnorm{O_1} \opnorm{O_2} \left<\xi_N, \cU_N^*(t;0)  \left(\cN + 1 \right)   \cU_N(t;0) \xi_N\right>\\
		&
		\leq N \hsnorm{O_1} \opnorm{O_2}  \norm{(\cN + 1 )^\frac{1}{2}\cU_N(t;0) \xi_N}.
	\end{align*}
	\begin{align*}
		&|A_{13}|\\
		& = \bigg| \int \dx_1 \dx_2 \mathrm{d}z_1 \mathrm{d}z_2\ O_1(x_1;z_1) O_2(x_2;z_2)  \big<  \xi_N, \cU^*_N(t;0) a^*(\bar{\rmv}_{t,z_1}) a(\bar{\rmv}_{t,x_1})a^*(\bar{\rmv}_{t,z_2})a(\bar{\rmv}_{t,x_2}) \cU_N(t;0)  \xi_N \big> \bigg|\\
		&
		\leq \opnorm{\bar{\rmv}_t \overline{O^*_1} {\rmv}_t} \opnorm{\bar{\rmv}_t \overline{O^*_2} {\rmv}_t}  \left<\xi_N, \cU_N^*(t;0)  \cN^2 \cU_N(t;0) \xi_N\right>\\
		&
		\leq \hsnorm{O_1} \opnorm{O_2} \left<\xi_N, \cU_N^*(t;0) \cN^2   \cU_N(t;0) \xi_N\right>\\
		&
		\leq N \hsnorm{O_1} \opnorm{O_2}  \norm{\cN \cU_N(t;0) \xi_N}.
	\end{align*}
	\begin{align*}
		&|A_{14}|\\
		& = \bigg| \int \dx_1 \dx_2 \mathrm{d}z_1 \mathrm{d}z_2\ O_1(x_1;z_1) O_2(x_2;z_2)  \big<  \xi_N, \cU^*_N(t;0) a^*(\bar{\rmv}_{t,z_1}) a(\bar{\rmv}_{t,x_1}) a(\rmu_{t,z_2}) a(\bar{\rmv}_{t,x_2}) \cU_N(t;0)  \xi_N \big> \bigg|\\
		&
		\leq \hsnorm{\bar{\rmv}_t \overline{O^*_1} {\rmv}_t} \hsnorm{\bar{\rmu}_t \overline{O^*_2} {\rmv}_t}  \left<\xi_N, \cU_N^*(t;0) \cN \cU_N(t;0) \xi_N\right>^\frac{1}{2} \left<\xi_N, \cU_N^*(t;0) \cN \cU_N(t;0) \xi_N\right>^\frac{1}{2}\\
		&
		\leq \sqrt{N} \hsnorm{O_1} \opnorm{O_2} \left<\xi_N, \cU_N^*(t;0) \cN   \cU_N(t;0) \xi_N\right>\\
		&
		\leq N \hsnorm{O_1} \opnorm{O_2}  \norm{\cN^\frac{1}{2} \cU_N(t;0) \xi_N}.
	\end{align*}
	\begin{align*}
		&|A_{15}|\\
		& = \bigg| \int \dx_1 \dx_2 \mathrm{d}z_1 \mathrm{d}z_2\ O_1(x_1;z_1) O_2(x_2;z_2)  \big<  \xi_N, \cU^*_N(t;0) a^*(\rmu_{t,x_1})a^*(\bar{\rmv}_{t,z_1}) a^*(\rmu_{t,x_2}) a(\rmu_{t,z_2}) \cU_N(t;0)  \xi_N \big> \bigg|\\
		&
		\leq \hsnorm{\rmu_t O_1 \bar{\rmv}_t} \opnorm{{\rmu}_t O_2 {\rmu}_t} \left<\xi_N, \cU_N^*(t;0) \cN^2 \cU_N(t;0) \xi_N\right>^\frac{1}{2} \left<\xi_N, \cU_N^*(t;0) \cN \cU_N(t;0) \xi_N\right>^\frac{1}{2}\\
		&
		\leq \hsnorm{O_1} \opnorm{O_2} \left<\xi_N, \cU_N^*(t;0) \cN^2   \cU_N(t;0) \xi_N\right>\\
		&
		\leq N \hsnorm{O_1} \opnorm{O_2}  \norm{\cN \cU_N(t;0) \xi_N}.
	\end{align*}
	\begin{align*}
		&|A_{16}|\\
		& = \bigg| \int \dx_1 \dx_2 \mathrm{d}z_1 \mathrm{d}z_2\ O_1(x_1;z_1) O_2(x_2;z_2)  \big<  \xi_N, \cU^*_N(t;0)  a(\rmu_{t,z_1}) a(\bar{\rmv}_{t,x_1}) a^*(\bar{\rmv}_{t,z_2}) a(\bar{\rmv}_{t,x_2}) \cU_N(t;0)  \xi_N \big> \bigg|\\
		&
		\leq \hsnorm{\bar{\rmu}_t  \overline{O^*_1} \bar{\rmv}_t} \opnorm{\bar{\rmv}_t \overline{O^*_1} \rmv_t}  \left<\xi_N, \cU_N^*(t;0) \cN^2 \cU_N(t;0) \xi_N\right>^\frac{1}{2} \left<\xi_N, \cU_N^*(t;0) (\cN +1)\cU_N(t;0) \xi_N\right>^\frac{1}{2}\\
		&
		\leq \hsnorm{O_1} \opnorm{O_2} \left<\xi_N, \cU_N^*(t;0)\left(\cN + 1 \right)^2   \cU_N(t;0) \xi_N\right>\\
		&
		\leq N \hsnorm{O_1} \opnorm{O_2}  \norm{\cN \cU_N(t;0) \xi_N}.
	\end{align*}
	Additionally, we have
	\begin{align*}
		&|B_2|\\
		&= \bigg| \int \dx_1 \dx_2 \mathrm{d}z_1 \mathrm{d}z_2\ O_1(x_1;z_1) O_2(x_2;z_2)  \big<  \xi_N, \cU^*_N(t;0) , \left<\bar{\rmv}_{t,x_1}, \bar{\rmv}_{t,z_1} \right> a(\bar{\rmv}_{t,x_2}) a(\rmu_{t,z_2})   \cU_N(t;0)  \xi_N \big> \bigg|\\
		&
		= \bigg| \int \dx_1 \dx_2 \mathrm{d}z_1 \mathrm{d}z_2\ O_1(x_1;z_1) O_2(x_2;z_2)  \big<  \xi_N, \cU^*_N(t;0) , \overline{\omega_{N,t}(x_1;z_1) } \\
		&
		\qquad \int  \deta \deta'\ a_{\eta} a_{\eta'} \rmv_t(\eta;x_2) \overline{\rmu_t(\eta'; z_2)}   \cU_N(t;0)  \xi_N \big> \bigg| \\
		&
		= \bigg| \int \dx_1 \dx_2 \mathrm{d}z_1 \mathrm{d}z_2\    \big<  \xi_N, \cU^*_N(t;0) , O_1(x_1;z_1) \omega_{N,t}(z_1;x_1)  \\
		&
		\qquad \int  \deta \deta'\ a_{\eta} a_{\eta'} \rmv_t(\eta;x_2) O_2(x_2;z_2) \rmu_t(z_2;\eta')  \cU_N(t;0)  \xi_N \big> \bigg| \\
		&
		= \bigg| \int  \deta \deta'  \big<  \xi_N, \cU^*_N(t;0) , \left(\int \dx_1\  \big(O_1\omega_{N,t}\big)(x_1;x_1) \right) a_{\eta} a_{\eta'} \big(\rmv_t O_2 \rmu_t\big)(\eta;\eta')  \cU_N(t;0)  \xi_N \big> \bigg| \\
		&
		\leq \trnorm{O_1 \omega_{N,t}} \hsnorm{\rmv_t O_2 \rmu_t}  \norm{\cN^{1/2} \cU_N(t;0)  \xi_N}\\
		&
		\leq \hsnorm{O_1} \hsnorm{\omega_{N,t}} \hsnorm{\rmv_t} \opnorm{O_2}  \opnorm{\rmu_t}  \norm{\cN^{1/2} \cU_N(t;0)  \xi_N}\\
		&
		\leq N  \hsnorm{O_1}  \opnorm{O_2}  \norm{\cN^{1/2} \cU_N(t;0)  \xi_N},
	\end{align*}
	where we use the fact that
	\[
	\left| \Tr O_1 \omega_{N,t} \right|\leq \trnorm{O_1 \omega_{N,t}} \leq \hsnorm{O_1} \hsnorm{\omega_{N,t}} \leq \sqrt{N} \hsnorm{O_1}.
	\]
	\begin{align*}
		&|B_3|\\
		&
		= \bigg| \int \dx_1 \dx_2 \mathrm{d}z_1 \mathrm{d}z_2\ O_1(x_1;z_1) O_2(x_2;z_2)  \big<  \xi_N, \cU^*_N(t;0) , \left<\rmu_{t,z_1}, \rmu_{t,x_2} \right> a(\bar{\rmv}_{t,x_1}) a(\rmu_{t,z_2})   \cU_N(t;0)  \xi_N \big> \bigg|\\
		&
		= \bigg| \int \dx_1 \dx_2 \mathrm{d}z_1 \mathrm{d}z_2\ O_1(x_1;z_1) O_2(x_2;z_2)  \big<  \xi_N, \cU^*_N(t;0) , \rmu_t(z_1;x_2)\\
		&
		\qquad  \int  \deta \deta'\ a_{\eta} a_{\eta'} \rmv_t(\eta;x_1) \overline{\rmu_t(\eta'; z_2)}  \cU_N(t;0)  \xi_N \big> \bigg|\\
		&
		= \bigg| \int \dx_1 \dx_2 \mathrm{d}z_1 \mathrm{d}z_2\int  \deta \deta'\big<  \xi_N, \cU^*_N(t;0) , a_{\eta} a_{\eta'} \\
		&
		\qquad  \rmv_t(\eta;x_1) O_1(x_1;z_1 )\rmu_t(z_1;x_2) O_2(x_2;z_2) {\rmu_t(z_2;\eta')}  \cU_N(t;0)  \xi_N \big> \bigg|\\
		&
		= \bigg| \int  \deta \deta'\big<  \xi_N, \cU^*_N(t;0) , a_{\eta} a_{\eta'}  \big(\rmv_t O_1 \rmu_t O_2 \rmu_t\big)(\eta;\eta')  \cU_N(t;0)  \xi_N \big> \bigg|\\
		&
		\leq \hsnorm{\rmv_t O_1 \rmu_t O_2 \rmu_t}  \norm{\cN^{1/2} \cU_N(t;0)  \xi_N} \\
		&
		\leq \hsnorm{O_1} \opnorm{O_2}  \norm{\cN^{1/2} \cU_N(t;0)  \xi_N}.
	\end{align*}
	\begin{align*}
		&|B_4|\\
		&
		= \bigg| \int \dx_1 \dx_2 \mathrm{d}z_1 \mathrm{d}z_2\ O_1(x_1;z_1) O_2(x_2;z_2)  \big<  \xi_N, \cU^*_N(t;0) , \left<\rmu_{t,z_1}, \rmu_{t,x_2} \right> a^*(\bar{\rmv}_{t,z_2})  a(\bar{\rmv}_{t,x_1})  \cU_N(t;0)  \xi_N \big> \bigg|\\
		&
		\leq \hsnorm{\rmv_t O_1 \rmu_t O_2 \bar{\rmv}_t} \norm{\cN^{1/2} \cU_N(t;0)  \xi_N} \\
		&
		\leq \hsnorm{O_1} \opnorm{O_2}  \norm{\cN^{1/2} \cU_N(t;0)  \xi_N}.
	\end{align*}
	\begin{align*}
		&|B_5|\\
		&
		= \bigg| \int \dx_1 \dx_2 \mathrm{d}z_1 \mathrm{d}z_2\ O_1(x_1;z_1) O_2(x_2;z_2)  \big<  \xi_N, \cU^*_N(t;0) , \left<\rmu_{t,z_1}, \rmu_{t,x_2} \right> \left<\bar{\rmv}_{t,x_1}, \bar{\rmv}_{t,z_2} \right> \cU_N(t;0)  \xi_N \big> \bigg|\\
		&
		= \bigg| \int \dx_1 \dx_2 \mathrm{d}z_1 \mathrm{d}z_2\ O_1(x_1;z_1) O_2(x_2;z_2)  \big<  \xi_N, \cU^*_N(t;0) , \rmu_t(z_1;x_2) \omega_{N,t}(z_2;x_1) \cU_N(t;0)  \xi_N \big> \bigg|\\
		&
		= \bigg|  \big<  \xi_N, \cU^*_N(t;0) , \int \dx_1 \dx_2   \big(O_1 \rmu_t\big)(x_1;x_2)\big( O_2 \omega_{N,t}\big)(x_2;x_1) \cU_N(t;0)  \xi_N \big> \bigg|\\
		&
		= \bigg|  \big<  \xi_N, \cU^*_N(t;0) , \int \dx_1   \big(O_1 \rmu_t O_2 \omega_{N,t}\big)(x_1;x_1) \cU_N(t;0)  \xi_N \big> \bigg|\\
		&
		\leq \hsnorm{ O_1 \rmu_t O_2} \hsnorm{\omega_{N,t}} \\
		&
		\leq \sqrt{ N} \hsnorm{O_1} \opnorm{O_2} .
	\end{align*}
	\begin{align*}
		&|B_6|\\
		&
		= \bigg| \int \dx_1 \dx_2 \mathrm{d}z_1 \mathrm{d}z_2\ O_1(x_1;z_1) O_2(x_2;z_2)  \big<  \xi_N, \cU^*_N(t;0) , \left<\bar{\rmv}_{t,x_1}, \bar{\rmv}_{t,z_1} \right> a^*(\rmu_{t,x_2}) a(\rmu_{t,z_2}) \cU_N(t;0)  \xi_N \big>   \bigg|\\
		&
		= \bigg| \int \dx_1 \dx_2 \mathrm{d}z_1 \mathrm{d}z_2\ O_1(x_1;z_1) O_2(x_2;z_2)  \big<  \xi_N, \cU^*_N(t;0) ,\omega_{N,t}(z_1;x_1) \\
		&
		\qquad  \int  \deta \deta'\ a^*_{\eta} a_{\eta'} \rmu_t(\eta;x_2) \overline{\rmu_t(\eta'; z_2)} \cU_N(t;0)  \xi_N \big> \bigg| \\
		&
		= \bigg| \int  \deta \deta'\big<  \xi_N, \cU^*_N(t;0) , \int \dd{x_1} \big(O_1\omega_{N,t}\big) (x_1;x_1) a^*_{\eta} a_{\eta'} \big(\rmu_t {O_2} \rmu_t\big)(\eta;\eta') \cU_N(t;0)  \xi_N \big> \bigg| \\
		&
		\leq \hsnorm{O_1} \hsnorm{\omega_{N,t}} \opnorm{\rmu_t {O_2} \rmu_t}  \norm{\xi_N} \norm{\cN^{1/2} \cU_N(t;0)  \xi_N}\\
		&
		\leq \sqrt{N} \hsnorm{O_1} \opnorm{O_2}  \norm{\cN^{1/2} \cU_N(t;0)  \xi_N}
	\end{align*}
	\begin{align*}
		&|B_7|\\
		&
		= \bigg| \int \dx_1 \dx_2 \mathrm{d}z_1 \mathrm{d}z_2\ O_1(x_1;z_1) O_2(x_2;z_2)  \big<  \xi_N, \cU^*_N(t;0) , \left<\bar{\rmv}_{t,x_1}, \bar{\rmv}_{t,z_1} \right> a^*(\rmu_{t,x_2})  a^*(\bar{\rmv}_{t,z_2})  \cU_N(t;0)  \xi_N \big>   \bigg|\\
		&
		= \bigg| \int \dx_1 \dx_2 \mathrm{d}z_1 \mathrm{d}z_2\ O_1(x_1;z_1) O_2(x_2;z_2)  \big<  \xi_N, \cU^*_N(t;0) ,\omega_{N,t}(z_1;x_1) \\
		&
		\qquad  \int  \deta \deta'\ a^*_{\eta} a^*_{\eta'} \rmu_t(\eta;x_2) \overline{\rmv_t(\eta'; z_2)} \cU_N(t;0)  \xi_N \big> \bigg| \\
		&
		\leq \hsnorm{O_1} \hsnorm{\omega_{N,t}} \hsnorm{\rmu_t O_2 \bar{\rmv}_t}  \norm{\xi} \norm{(\cN+1)^{1/2} \cU_N(t;0)  \xi_N}\\
		&
		\leq \sqrt{N} \hsnorm{O_1} \opnorm{\rmu_t}\opnorm{O_2} \hsnorm{\rmv_t} \norm{\xi_N} \norm{(\cN+1)^{1/2} \cU_N(t;0)  \xi_N}\\
		&
		\leq N \hsnorm{O_1} \opnorm{O_2} \norm{(\cN+1)^{1/2} \cU_N(t;0)  \xi_N}.
	\end{align*}
	\begin{align*}
		&|B_{8}|\\
		&
		= \bigg| \int \dx_1 \dx_2 \mathrm{d}z_1 \mathrm{d}z_2\ O_1(x_1;z_1) O_2(x_2;z_2)  \big<  \xi_N, \cU^*_N(t;0) , \left<\bar{\rmv}_{t,x_2}, \bar{\rmv}_{t,z_1} \right>  a^*(\rmu_{t,x_1})a(\rmu_{t,z_2})   \cU_N(t;0)  \xi_N \big>   \bigg|\\
		&
		= \bigg| \int \dx_1 \dx_2 \mathrm{d}z_1 \mathrm{d}z_2\ O_1(x_1;z_1) O_2(x_2;z_2)  \big<  \xi_N, \cU^*_N(t;0) ,\omega_{N,t}(z_1;x_2) \\
		&
		\qquad  \int  \deta \deta'\ a^*_{\eta} a_{\eta'} \rmu_t(\eta;x_1) \overline{\rmu_t(\eta'; z_2)} \cU_N(t;0)  \xi_N \big> \bigg| \\
		&
		\leq \opnorm{\rmu_1 O_1 \omega_{N,t} O_2 \rmu_t} \norm{\xi} \norm{\cN \cU_N(t;0)  \xi_N}\\
		&
		\leq \opnorm{O_1} \opnorm{O_2} \opnorm{ \omega_{N,t}}  \norm{\xi_N} \norm{\cN \cU_N(t;0)  \xi_N}\\
		&
		\leq \hsnorm{O_1} \opnorm{O_2} \norm{\xi} \norm{\cN \cU_N(t;0)  \xi_N}.
	\end{align*}
	\begin{align*}
		&|B_{9}|\\
		&
		= \bigg| \int \dx_1 \dx_2 \mathrm{d}z_1 \mathrm{d}z_2\ O_1(x_1;z_1) O_2(x_2;z_2)  \big<  \xi_N, \cU^*_N(t;0) , \left<\bar{\rmv}_{t,x_2}, \bar{\rmv}_{t,z_1} \right> a^*(\rmu_{t,x_1})a^*(\bar{\rmv}_{t,z_2})   \cU_N(t;0)  \xi_N \big>   \bigg|\\
		&
		= \bigg| \int \dx_1 \dx_2 \mathrm{d}z_1 \mathrm{d}z_2\ O_1(x_1;z_1) O_2(x_2;z_2)  \big<  \xi_N, \cU^*_N(t;0) ,\omega_{N,t}(z_1;x_2) \\
		&
		\qquad  \int  \deta \deta'\ a^*_{\eta} a^*_{\eta'} \rmu_t(\eta;x_1) \overline{\rmv_t(\eta'; z_2)} \cU_N(t;0)  \xi_N \big> \bigg| \\
		&
		\leq \hsnorm{\rmu_t O_1 \overline{\omega}_{N,t} O_2 \bar{\rmv}_t} \norm{\xi_N} \norm{(\cN+1)^{1/2} \cU_N(t;0)  \xi_N}\\
		&
		\leq \hsnorm{O_1} \opnorm{O_2} \norm{(\cN+1)^{1/2} \cU_N(t;0)  \xi_N}.
	\end{align*}
	\begin{align*}
		&|B_{10}|\\
		&
		= \bigg| \int \dx_1 \dx_2 \mathrm{d}z_1 \mathrm{d}z_2\ O_1(x_1;z_1) O_2(x_2;z_2)  \big<  \xi_N, \cU^*_N(t;0) , \left<\rmu_{t,z_1}, \rmu_{t,x_2} \right> a^*(\rmu_{t,x_1})a(\rmu_{t,z_2})  \cU_N(t;0)  \xi_N \big>   \bigg|\\
		&
		= \bigg| \int \dx_1 \dx_2 \mathrm{d}z_1 \mathrm{d}z_2\ O_1(x_1;z_1) O_2(x_2;z_2)  \big<  \xi_N, \cU^*_N(t;0) , \rmu_t(z_1;x_2) \\
		&
		\qquad  \int  \deta \deta'\ a^*_{\eta} a_{\eta'} \rmu_t(\eta;x_1) \overline{\rmu_t(\eta'; z_2)} \cU_N(t;0)  \xi_N \big> \bigg| \\
		&
		\leq \hsnorm{\rmu_t O_1 \rmu_t O_2 \rmu_t}\norm{\xi_N} \norm{\cN^{1/2} \cU_N(t;0)  \xi_N}\\
		&
		\leq \hsnorm{O_1} \opnorm{O_2}  \norm{\cN^{1/2} \cU_N(t;0)  \xi_N}.
	\end{align*}
	\begin{align*}
		&|B_{11}|\\
		&
		= \bigg| \int \dx_1 \dx_2 \mathrm{d}z_1 \mathrm{d}z_2\ O_1(x_1;z_1) O_2(x_2;z_2)  \big<  \xi_N, \cU^*_N(t;0) ,  \left<\rmu_{t,z_1},\rmu_{t,x_2} \right> a^*(\rmu_{t,x_1})a^*(\bar{\rmv}_{t,z_2})  \cU_N(t;0)  \xi_N \big>   \bigg|\\
		&
		= \bigg| \int \dx_1 \dx_2 \mathrm{d}z_1 \mathrm{d}z_2\ O_1(x_1;z_1) O_2(x_2;z_2)  \big<  \xi_N, \cU^*_N(t;0) , \rmu_t(z_1;x_2) \\
		&
		\qquad  \int  \deta \deta'\ a^*_{\eta} a^*_{\eta'} \rmu_t(\eta;x_1) \overline{\rmv_t(\eta'; z_2)} \cU_N(t;0)  \xi_N \big> \bigg| \\
		&
		\leq \hsnorm{\rmu_t O_1 \rmu_t O_2 \bar{\rmv}_t} \norm{\xi_N} \norm{(\cN+1)^{1/2} \cU_N(t;0)  \xi_N}\\
		&
		\leq \hsnorm{O_1} \opnorm{O_2} \norm{(\cN+1)^{1/2} \cU_N(t;0)  \xi_N}.
	\end{align*}
	\begin{align*}
		&|B_{12}|\\
		&
		= \bigg| \int \dx_1 \dx_2 \mathrm{d}z_1 \mathrm{d}z_2\ O_1(x_1;z_1) O_2(x_2;z_2)  \big<  \xi_N, \cU^*_N(t;0) ,   \left<\bar{\rmv}_{t,x_2}, \bar{\rmv}_{t,z_2}\right> a^*(\bar{\rmv}_{t,z_1}) a(\bar{\rmv}_{t,x_1}) \cU_N(t;0)  \xi_N \big>   \bigg|\\
		&
		= \bigg| \int \dx_1 \dx_2 \mathrm{d}z_1 \mathrm{d}z_2\ O_1(x_1;z_1) O_2(x_2;z_2)  \big<  \xi_N, \cU^*_N(t;0) ,\omega_{N,t}(z_2;x_2) \\
		&
		\qquad  \int  \deta \deta'\ a^*_{\eta} a_{\eta'} \overline{\rmv_t(\eta;z_1)} \rmv_t(\eta'; x_1) \cU_N(t;0)  \xi_N \big> \bigg| \\
		&
		\leq \opnorm{O_2} \trnorm{\omega_{N,t}} \opnorm{\bar{\rmv}_t \overline{O^*_1}  \rmv_t}  \norm{\xi_N} \norm{\cN \cU_N(t;0)  \xi_N}\\
		&
		\leq N \hsnorm{O_1} \opnorm{O_2}   \norm{\cN  \cU_N(t;0)  \xi_N}.
	\end{align*}
	\begin{align*}
		&|B_{13}|\\
		&
		= \bigg| \int \dx_1 \dx_2 \mathrm{d}z_1 \mathrm{d}z_2\ O_1(x_1;z_1) O_2(x_2;z_2)  \big<  \xi_N, \cU^*_N(t;0) ,   \left<\bar{\rmv}_{t,x_2}, \bar{\rmv}_{t,z_1}\right> a^*(\bar{\rmv}_{t,z_2}) a(\bar{\rmv}_{t,x_1}) \cU_N(t;0)  \xi_N \big>   \bigg|\\
		&
		= \bigg| \int \dx_1 \dx_2 \mathrm{d}z_1 \mathrm{d}z_2\ O_1(x_1;z_1) O_2(x_2;z_2)  \big<  \xi_N, \cU^*_N(t;0) ,\omega_{N,t}(z_1;x_2) \\
		&
		\qquad  \int  \deta \deta'\ a^*_{\eta} a_{\eta'} \overline{\rmv_t(\eta;z_2)} \rmv_t(\eta'; x_1) \cU_N(t;0)  \xi_N \big> \bigg| \\
		&
		\leq  \hsnorm{\bar{\rmv}_t \overline{O^*_2} \overline{\omega}_{N,t} \overline{O^*_1}  \rmv_t}  \norm{\xi_N} \norm{\cN^{1/2} \cU_N(t;0)  \xi_N}\\
		&
		\leq \hsnorm{O_1} \opnorm{O_2}  \norm{\cN^{1/2} \cU_N(t;0)  \xi_N}.
	\end{align*}
	\begin{align*}
		&|B_{14}|\\
		&
		= \bigg| \int \dx_1 \dx_2 \mathrm{d}z_1 \mathrm{d}z_2\ O_1(x_1;z_1) O_2(x_2;z_2)  \big<  \xi_N, \cU^*_N(t;0) ,   \left< \bar{\rmv}_{t,x_2}, \bar{\rmv}_{t,z_2} \right> a^*(\rmu_{t,x_1}) a(\rmu_{t,z_1})  \xi_N \big>   \bigg|\\
		&
		= \bigg| \int \dx_1 \dx_2 \mathrm{d}z_1 \mathrm{d}z_2\ O_1(x_1;z_1) O_2(x_2;z_2)  \big<  \xi_N, \cU^*_N(t;0) ,\omega_{N,t}(z_2;x_2) \\
		&
		\qquad  \int  \deta \deta'\ a^*_{\eta} a_{\eta'} \rmu_t(\eta; x_1)\overline{\rmu_t(\eta)} \cU_N(t;0)  \xi_N \big> \bigg| \\
		&
		\leq  \hsnorm{\bar{\rmv}_t \overline{O^*_2} \overline{\omega}_{N,t} \overline{O^*_1}  \rmv_t}  \norm{\xi_N} \norm{\cN^{1/2} \cU_N(t;0)  \xi_N}\\
		&
		\leq \hsnorm{O_1} \opnorm{O_2}  \norm{\cN^{1/2} \cU_N(t;0)  \xi_N}.
	\end{align*}
	\begin{align*}
		&|B_{15}|\\
		&
		= \bigg| \int \dx_1 \dx_2 \mathrm{d}z_1 \mathrm{d}z_2\ O_1(x_1;z_1) O_2(x_2;z_2)  \big<  \xi_N, \cU^*_N(t;0) ,    \left< \bar{\rmv}_{t,x_2}, \bar{\rmv}_{t,z_2} \right> a^*(\bar{\rmv}_{t,z_1}) a^*(\rmu_{t,x_1}) \cU_N(t;0)  \xi_N \big>   \bigg|\\
		&
		= \bigg| \int \dx_1 \dx_2 \mathrm{d}z_1 \mathrm{d}z_2\ O_1(x_1;z_1) O_2(x_2;z_2)  \big<  \xi_N, \cU^*_N(t;0) ,\omega_{N,t}(z_2;x_2) \\
		&
		\qquad  \int  \deta \deta'\ a^*_{\eta} a^*_{\eta'} \overline{\rmv_t(\eta;z_1)} \overline{\rmu_t(x_1;\eta')} \cU_N(t;0)  \xi_N \big> \bigg| \\
		&
		\leq \trnorm{O_2 \omega_{N,t}} \hsnorm{\bar{\rmv}_t \overline{O^*_1} \bar{\rmu}_t} \norm{\cN^\frac{1}{2} \cU_N(t;0) \xi_N}\\
		&
		\leq \hsnorm{O_1} \hsnorm{\omega_{N,t}} \hsnorm{\rmv_t} \opnorm{O_2} \norm{\cN^\frac{1}{2} \cU_N(t;0) \xi_N}\\
		&
		\leq N \hsnorm{O_1}\opnorm{O_2} \norm{\cN^\frac{1}{2} \cU_N(t;0) \xi_N}.
	\end{align*}
	\begin{align*}
		&|B_{16}|\\
		&
		= \bigg| \int \dx_1 \dx_2 \mathrm{d}z_1 \mathrm{d}z_2\ O_1(x_1;z_1) O_2(x_2;z_2)  \big<  \xi_N, \cU^*_N(t;0) ,   \left<\bar{\rmv}_{t,x_1}, \bar{\rmv}_{t,z_1}\right> a^*(\bar{\rmv}_{t,z_2}) a(\bar{\rmv}_{t,x_2})  \cU_N(t;0)  \xi_N \big>   \bigg|\\
		&
		= \bigg| \int \dx_1 \dx_2 \mathrm{d}z_1 \mathrm{d}z_2\ O_1(x_1;z_1) O_2(x_2;z_2)  \big<  \xi_N, \cU^*_N(t;0) ,\omega_{N,t}(z_1;x_1) \\
		&
		\qquad  \int  \deta \deta'\ a^*_{\eta} a_{\eta'} \overline{\rmv_t(\eta;z_2)} \rmv_t(\eta'; x_2) \cU_N(t;0)  \xi_N \big> \bigg| \\
		&
		\leq \hsnorm{O_1} \hsnorm{\omega_{N,t}} \opnorm{\bar{\rmv}_t \overline{O^*_2} \rmv_t} \norm{\xi_N} \norm{\cN \cU_N(t;0) \xi_N}\\
		&
		\leq  \sqrt{N} \hsnorm{O_1} \opnorm{O_2} \norm{ \cN \cU_N(t;0)  \xi_N}.
	\end{align*}

\vfill
\end{appendices}

\end{document}